\documentclass[journal]{IEEEtran}
\usepackage{lipsum}
\usepackage{color}
\ifCLASSINFOpdf
   \usepackage[pdftex]{graphicx}
   \usepackage{epstopdf}
   \graphicspath{{Images/}}
\else
\fi
\usepackage{amsthm}
\usepackage{mathtools}
\usepackage{cancel}
\usepackage{xcolor}
\usepackage{relsize}

\usepackage{amssymb}
\usepackage{blindtext}
\usepackage{amsmath}
\ifCLASSOPTIONcompsoc
\usepackage[caption=false,font=normalsize,labelfont=sf,textfont=sf]{subfig}
\else
\usepackage[caption=false,font=footnotesize]{subfig}
\fi

\usepackage{stfloats}

\usepackage{bm}
\usepackage{multirow}
\usepackage{multicol}
\usepackage{mathdots}
\newtheorem{theorem}{Theorem}
\newtheorem{assumption}{Assumption}
\newtheorem{corollary}{Corollary}[theorem]
\newtheorem{remark}{Remark}[theorem]

\newcommand{\blam}{\bm{\lambda}}
\newcommand{\bLam}{\bm{\Lambda}}
\newcommand{\bone}{\bm{1}}
\newcommand{\bE}{\bm{E}}
\newcommand{\bxi}{\bm{\xi}}
\newcommand{\bdelta}{\bm{\delta}}

\usepackage[square,numbers]{natbib}
\usepackage[belowskip=-15pt,aboveskip=0pt,small]{caption}
\begin{document}
\title{Isochronous Architecture-Based Voltage-Active Power Droop for Multi-Inverter Systems}
% \author{}
\author{Sourav~Patel,~\IEEEmembership{Student Member,~IEEE,}
        Soham~Chakraborty,~\IEEEmembership{Student Member,~IEEE,}
        Blake~Lundstrom,~\IEEEmembership{Senior Member,~IEEE,}
        Srinivasa~Salapaka,~\IEEEmembership{Senior Member,~IEEE,}
        and~Murti V. Salapaka,~\IEEEmembership{Fellow,~IEEE}% <-this % stops a space
        \vspace{-0.9cm}
\thanks{S. Patel, S. Chakraborty, M. Salapaka are with the Department of Electrical and Computer Engineering, University of Minnesota, Minneapolis, 55455 MN, USA~ e-mail: patel292@umn.edu, chakr138@umn.edu, murtis@umn.edu}% <-this % stops a space
\thanks{B. Lundstrom is with The National Renewable Energy Laboratory, Golden, 80401 CO, USA~% <-this
e-mail: blake.lundstrom@nrel.gov}
\thanks{S. Salapaka is with the Department of Mechanical Science and Engineering, University of
Illinois at Urbana-Champaign, 61801 IL, USA~% stops a space
e-mail: salapaka@illinois.edu}
% \thanks{Manuscript received April 19, 2005; revised August 26, 2015.}
\thanks{The authors acknowledge the support of ARPA-E for supporting this research through the project titled ‘A Robust Distributed Framework for Flexible Power Grids’ via grant no. DE-AR0000701.}
}
\maketitle
\begin{abstract}
Advanced microgrids consisting of distributed energy resources interfaced with multi-inverter systems are becoming more common. Consequently, the effectiveness of voltage and frequency regulation in microgrids using conventional droop-based methodologies is challenged by uncertainty in the size and schedule of loads. This article proposes an isochronous architecture of parallel inverters with only voltage-active power droop (VP-D) control for improving active power sharing as well as plug-and-play of multi-inverter based distributed energy resources (DERs). In spite of not employing explicit control for frequency regulation, this architecture allows even sharing of reactive power while maintaining reduced circulating currents between inverters. The performance is achieved even when there are mismatches between commanded reference and power demanded from the actual load in the network.
The isochronous architecture is implemented by employing a global positioning system (GPS) to disseminate the clock timing signals that enable the microgrid to maintain nominal system frequency in the entire network. This enables direct control of active power through voltage source inverter (VSI) output voltage regulation, even in the presence of system disturbances. A small signal eigenvalue analysis of a multi-inverter system near the steady-state operating point is presented to evaluate the stability of the multi-inverter system with the proposed VP-D control. Simulation studies and hardware experiments on an 1.2 kVA prototype are conducted. The effectiveness of the proposed architecture towards active and reactive power sharing between inverters with load scenarios are demonstrated. Results of the hardware experiments corroborate the viability of the proposed VP-D control architecture.
\end{abstract}
%%%%%%%%%%%%%%%%%%%%%%%%%%%%%%%%%%%%%%%%%%%%%%%%%%%%%%%%%%%%%%%%
%%%%%%%%%%%%%%%%%%%%%%%%%%%%%%%%%%%%%%%%%%%%%%%%%%%%%%%%%%%%%%%%
%%%%%%%%%%%%%%%%%%%%%%%%%%%%%%%%%%%%%%%%%%%%%%%%%%%%%%%%%%%%%%%%
%%%%%%%%%%%%%%%%%%%%%%%%%%%%%%%%%%%%%%%%%%%%%%%%%%%%%%%%%%%%%%%%
%%%%%%%%%%%%%%%%%%%%%%%IEEEkeywords%%%%%%%%%%%%%%%%%%%%%%%%%%%%%
%%%%%%%%%%%%%%%%%%%%%%%%%%%%%%%%%%%%%%%%%%%%%%%%%%%%%%%%%%%%%%%%
%%%%%%%%%%%%%%%%%%%%%%%%%%%%%%%%%%%%%%%%%%%%%%%%%%%%%%%%%%%%%%%%
%%%%%%%%%%%%%%%%%%%%%%%%%%%%%%%%%%%%%%%%%%%%%%%%%%%%%%%%%%%%%%%%
%%%%%%%%%%%%%%%%%%%%%%%%%%%%%%%%%%%%%%%%%%%%%%%%%%%%%%%%%%%%%%%%
%\markboth{IEEE TRANSACTIONS ON INDUSTRIAL ELECTRONICS}%
{}
\definecolor{limegreen}{rgb}{0.2, 0.8, 0.2}
\definecolor{forestgreen}{rgb}{0.13, 0.55, 0.13}
\definecolor{greenhtml}{rgb}{0.0, 0.5, 0.0}
\vspace{-0.55cm}
%%%%%%%%%%%%%%%%%%%%%%%%%%%%%%%%%%%%%%%%%%%%%%%%%%%%%%%%%%%%%%%%
%%%%%%%%%%%%%%%%%%%%%%%%%%%%%%%%%%%%%%%%%%%%%%%%%%%%%%%%%%%%%%%%
%%%%%%%%%%%%%%%%%%%%%%%%%%%%%%%%%%%%%%%%%%%%%%%%%%%%%%%%%%%%%%%%
%%%%%%%%%%%%%%%%%%%%%%%%%%%%%%%%%%%%%%%%%%%%%%%%%%%%%%%%%%%%%%%%
%%%%%%%%%%%%%%%%%%%%%%%%Introduction%%%%%%%%%%%%%%%%%%%%%%%%%%%%
%%%%%%%%%%%%%%%%%%%%%%%%%%%%%%%%%%%%%%%%%%%%%%%%%%%%%%%%%%%%%%%%
%%%%%%%%%%%%%%%%%%%%%%%%%%%%%%%%%%%%%%%%%%%%%%%%%%%%%%%%%%%%%%%%
%%%%%%%%%%%%%%%%%%%%%%%%%%%%%%%%%%%%%%%%%%%%%%%%%%%%%%%%%%%%%%%%
%%%%%%%%%%%%%%%%%%%%%%%%%%%%%%%%%%%%%%%%%%%%%%%%%%%%%%%%%%%%%%%%
\section{Introduction}
% It introduces the current SOA of DERs and microgrids
\IEEEPARstart{M}{odernization} of rapidly dispatchable DERs to provide demand response and ancillary services are transforming emergent microgrids into advanced microgrids. Advanced microgrids enable additional flexibility, resilience and reliability for local resources as well as support of the large scale grid when connected \cite{marzal_current_2018}. These microgrids can be operated in both grid-connected as well as islanded modes. In the islanded mode, DERs support local loads in the microgrid through either centralized, distributed or de-centralized control architectures. Controllers employed either locally at individual DERs or at a microgrid level are responsible for stabilizing microgrid voltage \cite{vasquez2009voltage}, maintaining power quality with plug-and-play capabilities \cite{blaabjerg2006overview}, apportioning load between various distributed generations (DGs) \cite{patel_distributed_2017}, and primary frequency response support \cite{lundstrom_fast_2018}. A number of hierarchical control schemes achieve one or more of these objectives \cite{guerrero_hierarchical_2011}.
\par A distributed and communication-less approach of controlling inverters is typically achieved through droop control methods. These methods use measurements obtained locally at the DERs and leverage the relationships between system voltage, line frequency, and power supplied by the DERs to ensure regulation of voltage and frequency in the network \cite{chandorkar_control_1993}. In microgrids with high line inductance (high $X/R$ ratio), DERs are generally controlled by conventional $P-f$ and $Q-V$ droop control \cite{guerrero2005output, yao2011design}, where $P$, $Q$, $f$, and $V$ are active power, reactive power, frequency, and voltage respectively. The droop strategies in microgrids are motivated by the the autonomous control of synchronous generators, with inherent large rotating inertia directly coupled with the network that primarily rely on conventional droop strategies for regulation of system voltage and frequency. However, low voltage (LV) microgrids use more converter-interfaced DERs, and thus, exhibit low natural inertia with significant $R/X$ ratio because of the small spatial expanse of the networks \cite{vandoorn2012automatic}. Consequently, non-conventional $P-V$ and $Q-f$ droop control strategies \cite{vandoorn2012automatic, sao2008control} have gained attention for LV microgrids with the advantages of direct voltage control \cite{engler_droop_2005} and improved harmonic power sharing \cite{guerrero_resistive_2007} compared to conventional droop control. Moreover, in order to improve transient performance and leverage decoupling compensations for droop controlled MMG, there are several works in the literature concentrating on modifying the droop characteristic equations \cite{droopnew1,droopnew2,droopnew3,droopnew4,droopnew5}. However, these  more complicated droop characteristic equations are difficult to  analyze, give stability guarantees and difficult to incorporate in applications.
\par In traditional droop based strategies, with the frequency related droop strategy employed, deviation of microgrid frequency caused by mismatches between load and generation is inevitable. Here, LV microgrids with low inertia converter-interfaced DERs, the frequency may deviate considerably from nominal  \cite{guerrero2005output,6200347,majumder_angle_2009}. Typically, a secondary controller, operating at a slower time-scale than the autonomous droop based response, utilizes communication between DERs to restore the steady state system frequency to nominal \cite{shafiee2014distributed}; however, it is difficult to mitigate the transient frequency deviation caused by sudden load changes especially in islanded microgrids. This regulation becomes even more difficult in the presence of uncertain renewable energy resources along with the added complexity of synchronizing DERs in a parallel inverter network. Other limitations in traditional droop methods arise from the trade-offs between accurate power sharing on one hand and the desire for stringent regulation of system voltage and frequency on the other. These trade-offs result from mismatches in the commanded reference power and actual load as well as the output impedance of VSI and line parameters resulting in large circulating currents.
% however, it is difficult to mitigate the transient frequency deviation caused by sudden load changes especially in islanded microgrids with higher renewable energy resources resulting in considerable transient frequency fluctuations even when a secondary control hierarchy is used \cite{golsorkhi2017gps}. Moreover, control of frequency adds more complexity with the added challenge of synchronization in a parallel inverter network. 
% \textcolor{blue}{Patel, This may not be the best place to add papers on droop and virtual impedances. This paper already mentioned secondary control and frequency deviation. I was wondering whether end of previous paragraph would be the place to mention couple of latest works on droop.}

\par In  this  article,  we  propose  using  only  VP-D  control  for voltage  regulation  and  address  frequency  regulation  by  providing  a  common  timing  signal  from  a  GPS  receiver  to  all the inverters. This enables generation of a common frequency reference for all the inverters and avoids the need for explicit frequency control; moreover, the framework obviates the need for   communication between inverters. In contrast  to  existing  GPS-based  methods  \cite{golsorkhi2017gps,riverso2015plug}, our approach simultaneously  achieves  multiple  objectives  including  reduced circulating  currents  between  VSIs  with  better  voltage  regulation,  reduced  complexity  due  to  the  absence  of  PLL  required  for  frequency  synchronization. Accordingly,  the  main  contributions  of  the proposed isochronous architecture are summarized as follows.\\
1)  Access  to  common  clock pulses in this architecture enables operation of the  Multi-inverter MicroGrid (MMG) at  a  single  frequency  even  in  the  presence  of system disturbances. VP-D control decouples voltage regulation objectives  from frequency variations thereby alleviating issues such as frequency dependent control resulting in poor voltage regulation and poor disturbance rejection. This architecture also enables improved  start  up  transients, reduced circulating currents between participating VSIs and easy plug-and-play of  VSIs in the MMG. \\
2) The   architecture   does   away   with   issues   stemming from   active   frequency-regulation   in   full-droop   methods.   VP-D  control  allows  direct  control  of  active  power  with regulation  of  voltage  magnitude.  The  resulting  control  law guarantees that the parallel inverters share reactive power such that  deviations  from  desired  reactive  power  sharing  remain bounded while serving a common load even without explicitly controlling the reactive power flow. The resulting controller design is simpler because no explicit $Q-f$ control  is  required.\\
3) The  proposed  framework  facilitates  a  comprehensive  analysis  that  provides: i) delineation analytically of how circulating current between VSIs depends on output voltage regulation error; ii)  small-signal  based  stability  analysis  of  the  ac  microgrid  by  having access to the common clock signal to generate $\alpha-\beta$ components without added complexities of PLL dynamics. These provide useful design  guidelines  for  selecting  droop  gains and virtual output resistances to maintain system stability\\
Results show that the proposed architecture exhibits improved  transients  over  full-droop  implementations  during start-up  and  different  plug-and-play  scenarios. For scenarios with symmetric active, reactive power sharing, an improvement of $96\% - 101\%$ is observed in the magnitude of circulating current between any two inverters, reactive power sharing ratio accuracy improved from $2.61\%-8.59\%$ to be within $0.68\%-1.07\%$ of desired and transient response time from 0.1238 s to 0.02 s. Voltage regulation error at PCC improved from $5.71\%$  to $< 0.01\%$ (of nominal) during plug-and-play of the MMG.
Experimental results  are  provided  for  a  system  of  two  0.6  kVA  inverters under  different  load  variations  and  sharing  ratios to  validate  the  feasibility  and  effectiveness  of  the  proposed isochronous architecture.\\
The rest of the paper is organized as follows. Section \ref{systemdesc} describes the system and control architecture. Section \ref{analysis} provides analytical results for the multi-inverter system, and Section \ref{results} presents simulation and experimental validation on a two-inverter microgrid network. Finally, Section \ref{conclusion} summarizes observations and insights that highlight this work. 

\begin{figure}[t]
\centering
\includegraphics[scale=0.7,trim={11.7cm 9cm 1.5cm 1cm},clip]{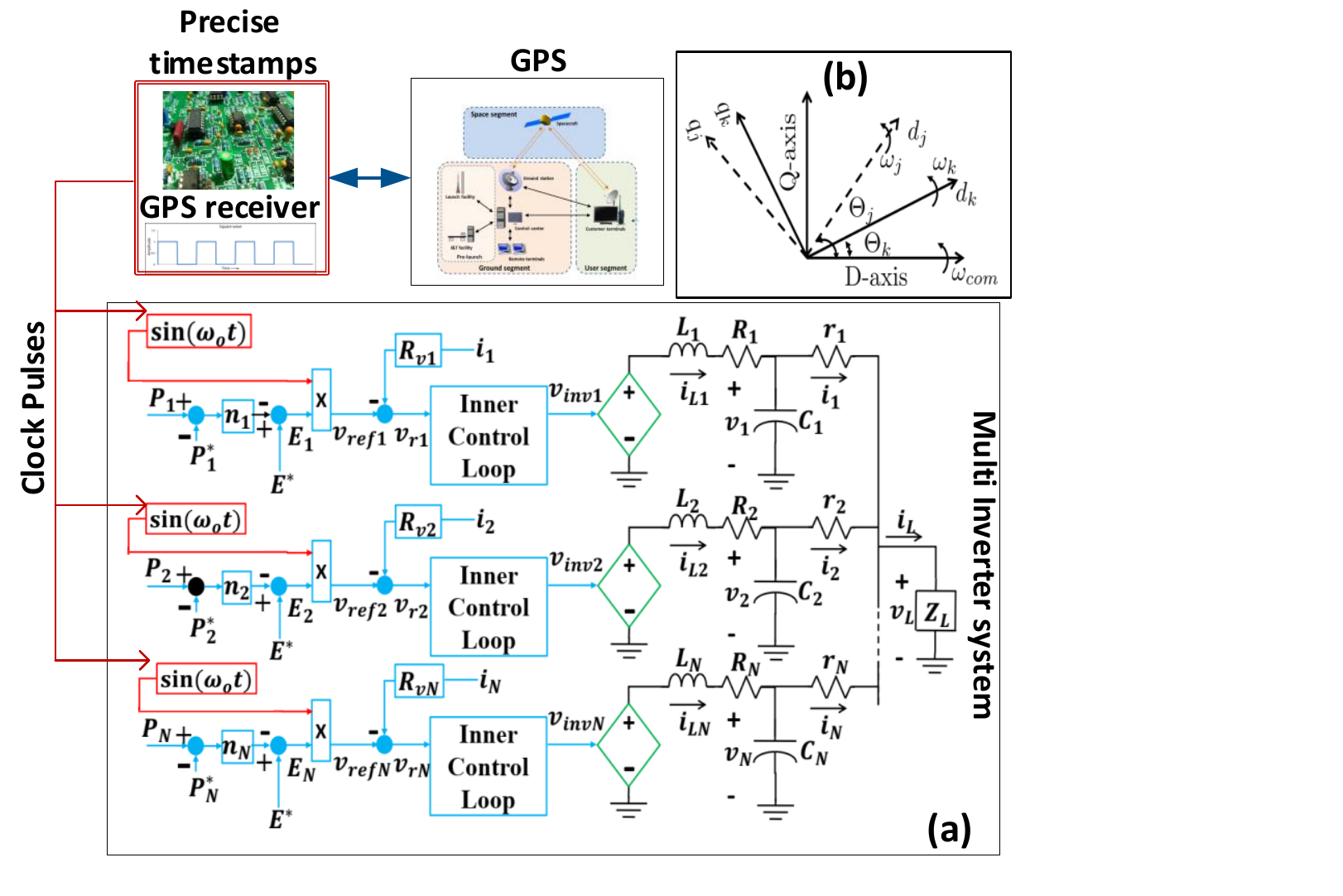}
\caption{(a) A parallel $N$-inverter MMG with outer-droop and inner multi-loop controllers receiving GPS signals \cite{packard1996gps} to disseminate clock signals for frequency regulation, (b) representation of VSI $j$ and VSI $k$ with respect to common reference frame.}\label{fig:multiinverter}
\end{figure}
\section{System Description and Controller Design}\label{systemdesc}
This section provides the details of an LV resistive MMG network containing $N$ parallel VSIs, serving a common load, $Z_L$, connected at a single point of common coupling (PCC) as shown in Fig. \ref{fig:multiinverter}. Each inverter in the MMG includes a communication module that receives synchronized clock pulses from the GPS segment \cite{packard1996gps} through a GPS receiver. These clock pulses are then used by the proposed outer VP-D control to provide a synchronized sinusoidal reference voltage signal to inner-loop controllers located at each VSI forming the isochronous system architecture. We begin our discussion by describing first the dynamics of the $k^{th}$ single inverter system and the controllers employing the proposed control law for achieving output voltage regulation and load disturbance rejection through a cascaded multi-loop control architecture. 
% All VSIs in the microgrid receive synchronized clock pulses from the GPS segment \cite{packard1996gps} through a GPS receiver or synchronize their local clocks, using zero-crossing detection and pulse counting, to a master clock signal transmitted as a square wave. This clock signal is then used to generate a synchronized sinusoidal reference voltage signal to inner-loop controllers forming the isochronous system architecture. 
%%%%%%%%%%%%%%%%%%%%%%%%%%%%%%%%%%%%%%%%%%%%%%%%%%%%%%%%%%%%%%%%
%%%%%%%%%%%%%%%%%%%Single Inverter System%%%%%%%%%%%%%%%%%%%%%%%
%%%%%%%%%%%%%%%%%%%%%%%%%%%%%%%%%%%%%%%%%%%%%%%%%%%%%%%%%%%%%%%%
%%%%%%%%%%%%%%%%%%%%%%%%%%%%%%%%%%%%%%%%%%%%%%%%%%%%%%%%%%%%%%%%
\subsection{Single Inverter System with multi-loop control design}\label{controller_design}
% \textcolor{blue}{ Can there be communication delays between different receivers? how are the signals synchronized at the receiving ends? Can we say that the geographical distances are small that the disparity in delays are negligible? it may be a good idea to say how small these delays can be (you can make a generic statement here and a more precise statement in Section C.}
A single phase H-bridge inverter with an output $LC$ filter is connected to load $Z_L$ through a line resistance $r_k$. The VSI is modeled as a controllable voltage source, $V_{invk}$, by employing an average model of the inverter \cite{yazdani_voltage-sourced_2010}. The dynamics of the $k^{th}$ VSI are described as:
\[L_k \frac{di_{Lk}}{dt} +R_ki_{Lk} = v_{inv k} - v_k, \hspace*{4mm} 
C_k \frac{dv_k}{dt} =  i_{Lk} - i_k\] 
where $v_{invk}$, $i_{Lk}$, $v_{k}$, and $i_{k}$ are average values over one switching cycle ($T_s$) of control signal, inductor current, output voltage, and output current, respectively (See Fig. \ref{fig:invertercontrol}). Taking a Laplace transform of these equations and eliminating $I_{Lk}(s)$, the open-loop averaged output voltage dynamics of the $k^{th}$ inverter are given as, $V_k(s)=$
\begin{align}\label{eq3} \scriptstyle
     \underbrace{\textstyle\frac{1}{L_k C_k s^2 +R_k C_k s +1}}_{G_{v_{invk}}}V_{invk}(s) - \scriptstyle\underbrace{\textstyle\frac{(L_k s+R_k)}{L_k C_k s^2 +R_k C_k s +1}}_{G_{ik}}I_k(s)
\end{align}
% where $G_{v_{invk}}:=1/(L_k C_k s^2 +R_k C_k s +1)$ and $G_{ik}:=(L_k s+R_k)/(L_k C_k s^2 +R_k C_k s +1)$.\\
% \begin{figure}[!t]\centering
% \includegraphics[scale=2,trim={0cm 0cm 0cm 0cm},clip]{singleInverter.pdf}
% \caption{Average model of $k^{th}$ H-bridge inverter with filter inductor $L_k$ with parasitic resistance $R_k$ and filter capacitor $C_k$ connected to load $Z_L$ through line resistance $r_k$}\label{singleinverter}
% \end{figure}
\begin{figure}[t]
\centering
\includegraphics[scale=0.36,trim={4.8cm 7.2cm 4.2cm 3.5cm},clip]{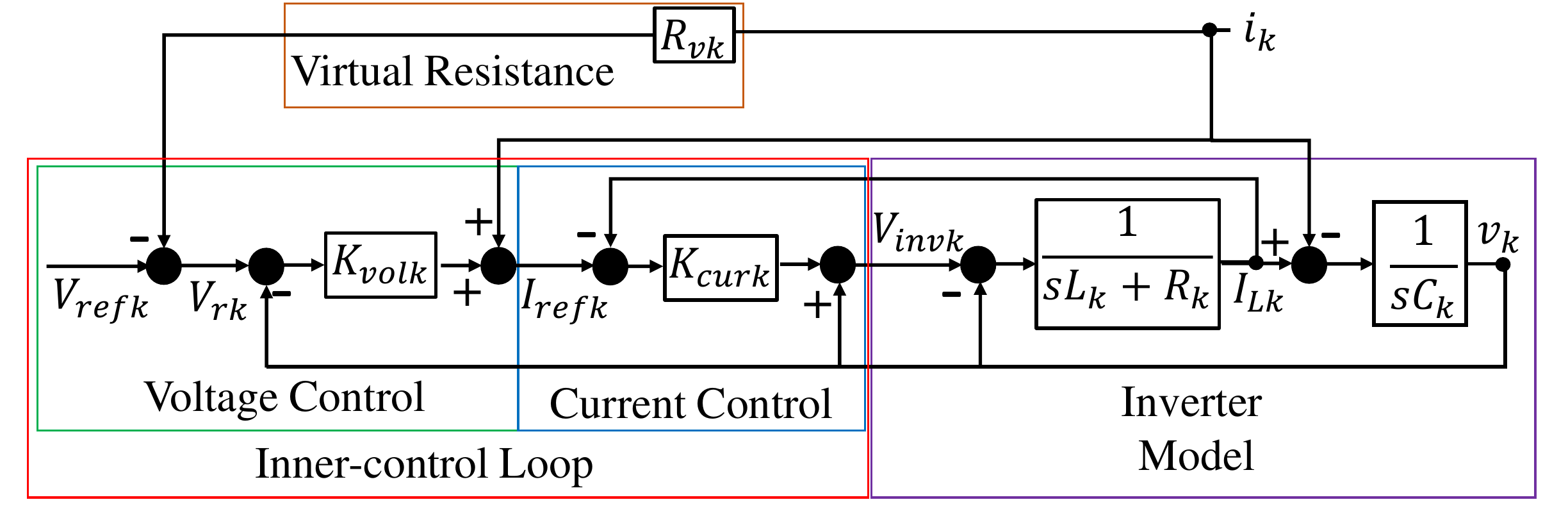}
\caption{The inner-control loop for voltage regulation with a virtual
resistance.}\label{fig:invertercontrol}
\end{figure}
% \begin{figure}[ht]
% \centering
% \includegraphics[scale=1,trim={3.5cm 0.3cm 7cm 0cm},clip]{equiInv.pdf}
% \caption{Thevenin equivalent circuit of $k^{th}$ inverter implemented with inner-control loop with virtual resistance as shown in Fig. \ref{fig:invertercontrol} }\label{fig:equiInv}
% \end{figure}
% \subsection{Inner Control loop Design}\label{controller_design}
% In this section, a control methodology is demonstrated for a single inverter system, as depicted in Fig. \ref{fig:invertercontrol}. 
The control design objectives are to: i) regulate the output voltage of the inverter ($V_{k}$) to track a reference signal ($V_{rk}$); ii) regulate the output inductor current of the inverter ($I_{Lk}$) to track a reference signal ($I_{refk}$); iii) reject disturbances due to load variations. We propose the following control law for the $k^{th}$ inverter:
\begin{align}\label{eq4}
V_{invk}(s) = K_{refk}V_{refk}(s) - K_{vk}V_k(s)  + K_{ik}I_k(s)
\end{align}
where $K_{refk}=K_{curk}K_{volk}$, $K_{vk}=K_{refk} + K_{curk}sC_k -1$, and $K_{ik}=-K_{refk}R_{vk}$.
As shown in Fig. \ref{fig:invertercontrol}, the control law takes the form of cascaded voltage and current controllers.The outer voltage controller, $K_{volk}$, acts on a sinusoidal reference signal provided by the outer droop law with the incorporation of a voltage drop, $i_kR_{vk}$, on a virtual resitance, $R_{vk}$. A feed-forward of the line current is used to minimize inrush currents and improve the transient response of the outer voltage controller. The output of the voltage controller, $I_{refk}$, provides the sinusoidal reference current signal to be tracked by the inductor current which generates the controllable voltage signal, $V_{invk}$, for the $k^{th}$ VSI. The main advantages of the inner-outer loop structure are the decoupled design of the fast inner current controller and the slow outer voltage controller as well as the increased robustness of the VSI to load variations.
% along with over-current protection for active current-limiting output under fault conditions.
By combining \eqref{eq3} and \eqref{eq4}, the closed-loop averaged output voltage dynamics of $k^{th}$ inverter are given by:
\begin{align}\label{eq5}
V_k(s) &= G_k\underbrace{(V_{refk}(s)-R_{vk}I_k(s))}_{V_{rk}(s)} - Z_kI_k(s) 
\end{align}
where $G_k=S_kG_{v_{invk}}K_{refk}$, $S_k=(1+G_{v_{invk}}K_{vk})^{-1}$ is the sensitivity transfer function,  $T_k=1-S_k$ is the complementary transfer function, and 
$Z_k(s)=S_kG_{ik}$. Note that, 
% where $G_k := (1 + G_{v_{invk}}K_{vk})^{-1}G_{v_{invk}}K_{refk}$ is the voltage gain (with complementary sensitivity transfer function, $\mathbf{T}_{V_{rk}\rightarrow V_k}(s)$, sensitivity transfer function, $\mathbf{S}_{V_{rk}\rightarrow V_k}(s):= 1 - \mathbf{T}_{V_{rk}\rightarrow V_k}(s)$), and 
$Z_k(j\omega_0) := (1 + G_{v_{invk}}K_{vk})^{-1}G_{ik}\vert_{s=j\omega_0}$ is the equivalent Thevenin output impedance of the $k^{th}$ inverter at the fundamental frequency, $\omega_0=2\pi60$ rad/s. 
% \textcolor{red}{$V_{refk}(s)$ is the sinusoidal reference voltage generated by the VP-D controller and isochronous $\sin(\omega_0 t)$ generation from a common clock.}
The primary control objectives thus translate to $G_k(j\omega_0)\approx 1$ for voltage regulation and the closed-loop inner- current loop transfer function from  $I_{refk}$ to $I_{Lk}$  be such that  $T_{I_{refk}\rightarrow I_{Lk}(j\omega_0)}\approx 1$ and $Z_k(j\omega_0)\ll 1$ for  rejecting the effects of load disturbances (See Fig. \ref{figfig}).\\
% Thus, the controller design should achieve, $\vert\mathbf{T}_{V_{rk}\rightarrow V_k}(s)\vert_{s=j\omega_0} \approx 1$, $\angle\mathbf{T}_{V_{rk}\rightarrow V_k}(s)\vert_{s=j\omega_0} \approx 0^{\circ}$ and $\vert\mathbf{T}_{I_{refk}\rightarrow I_{Lk}}(s)\vert_{s=j\omega_0} \approx 1$, $\angle\mathbf{T}_{I_{refk}\rightarrow I_{Lk}}(s)\vert_{s=j\omega_0} \approx 0^{\circ}$ for the voltage and current controllers respectively to meet the regulation objectives. Moreover, $\vert\mathbf{S}_{V_{rk}\rightarrow V_k}(s)\vert_{s=j\omega_0} \ll 1$, $\vert\mathbf{S}_{I_{refk}\rightarrow I_{Lk}}(s)\vert_{s=j\omega_0} \ll 1$ 
% and $\vert Z_k(s)\vert_{s=j\omega_0} \ll 1$ for ensuring good disturbance rejection.
 The current-control and voltage-control designs are based on the internal model principle \cite{yazdani_voltage-sourced_2010,li_control_2009}, where the  controllers $K_{curk}$ and $K_{volk}$ are designed to have poles  $s=\pm j\omega_0$ and loop-shaping-based design to cancel the  pole located at  $s=R_k/L_k$ by the current controller. Moreover, this design ensures proper separation in bandwidth of both the controllers (current control loop has the larger bandwidth), which enables the inner current loop to make the inductor current $I_{Lk}$ to  track the reference current $I_{refk}$.
%  the closed-loop bandwidth is sufficiently larger than the system frequency allowing for compensation from load variations while ensuring good tracking of $I_{refk}$. The bandwidth of the outer voltage control loop is designed to be sufficiently slower than the current control loop. This separation in bandwidth  allows for treating the output of the voltage controller as a reference signal to the inner-current loop.  
 
% Moreover, current control provides over-current protection for active current-limiting output under fault conditions.

 %Thus, by mapping $E\angle \phi$ and $V\angle 0$ to $|G_kV_{refk}(j\omega_0)|\angle G_kV_{refk}(j\omega_0)$ and $|V_{k}(j\omega_0)|\angle V_{k}(j\omega_0)$ respectively of \eqref{eq5}, and from Theorem \ref{thm:VfollowVREFandI}, the V-PD control can be achieved for purely resistive grid.
An output virtual resistance loop \cite{guerrero_resistive_2007} is also adopted to alleviate the sensitivity of active power sharing to the line and output impedances. This can be achieved by designing the virtual resistance, $R_{vk}$, to dominate the line impedance of the $k^{th}$ VSI, rendering the design insensitive to line and filter parameters. However, $R_{vk}$ is bounded above to ensure VSI output voltage regulation to be within prescribed limits as larger $R_{vk}$ results in larger regulation error by reducing $V_{rk}$. 
% Note that the output impedance of the $k^{th}$ inverter ($Z_k$) behaves as predominantly resistive in nature at system frequency $\omega_0$ and is equal to the value of the selected virtual resistance ($R_{vk}$). This is essential for implementing the VP-D control by making the overall system more damped. 
\begin{figure}[t]
\centering
\includegraphics[scale=0.18]{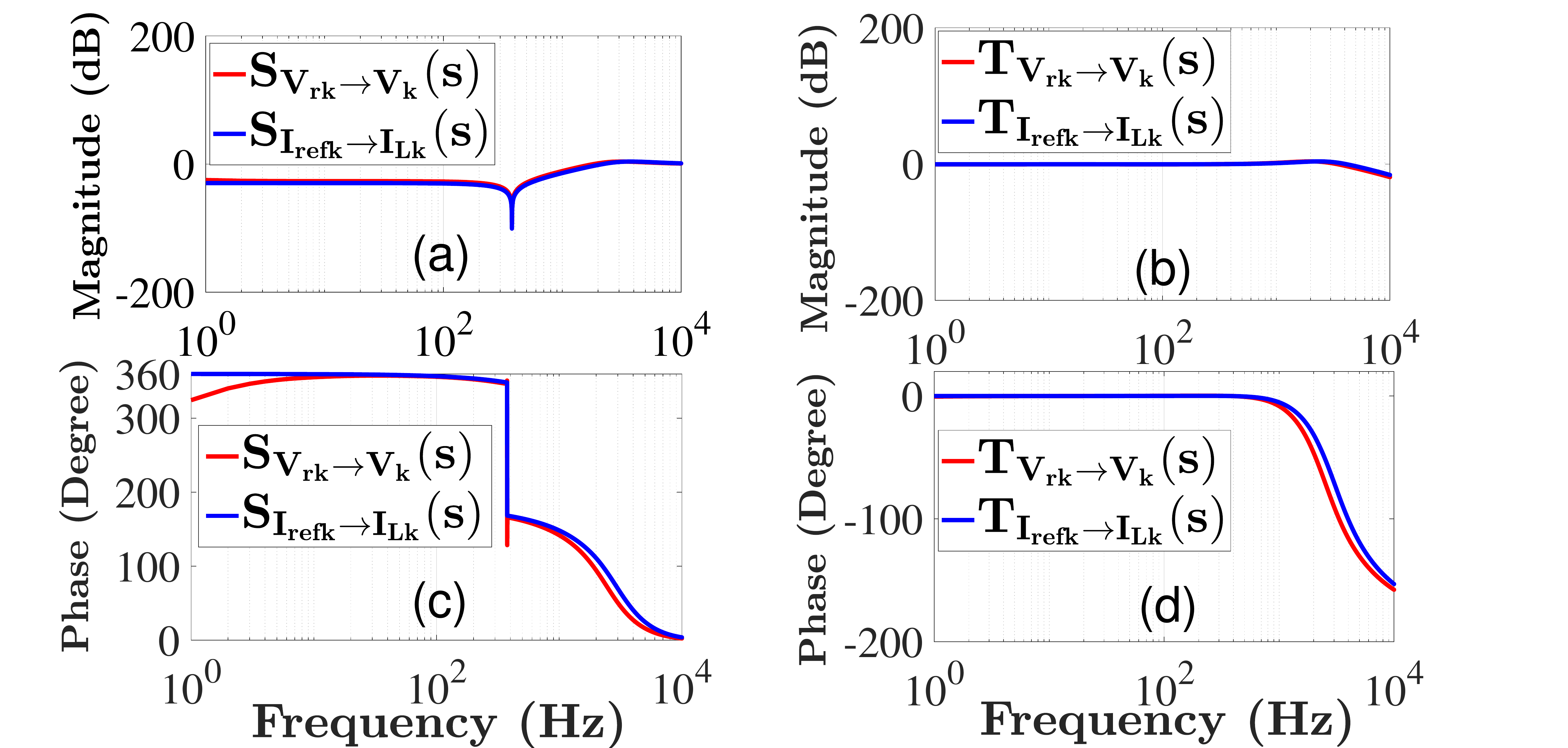}\label{fig:magSvSi}
% ~~
% \subfloat[]{\includegraphics[scale=0.12,trim={1.5cm 5.5cm 25cm 1.5cm},clip]{magTvTi.eps}\label{fig:magTvTi}}\\
% \subfloat[]{\includegraphics[scale=0.12,trim={2cm 5.5cm 25cm 1.5cm},clip]{phaseSvSi.eps}\label{fig:phaseSvSi}}
% ~~
% \subfloat[]{\includegraphics[scale=0.12,trim={1.5cm 5.5cm 25cm 1.5cm},clip]{phaseTvTi.eps}\label{fig:phaseTvTi}}
\caption{Bode plots of sensitivity and complementary sensitivity transfer function of voltage and current controllers.}
\label{figfig}
\end{figure}
%%%%%%%%%%%%%%%%%%%%%%%%%%%%%%%%%%%%%%%%%%%%%%%%%%%%%%%%%%%%%%%%
%%%%%%%%%%%%%%    Assumptions 1     %%%%%%%%%%%%%%%%%%%%%%%%%%%%%%%%
%%%%%%%%%%%%%%%%%%%%%%%%%%%%%%%%%%%%%%%%%%%%%%%%%%%%%%%%%%%%%%%%
\begin{assumption}\label{assmp:VfollowVREFandI}
The proposed control design shown in Fig. \ref{fig:invertercontrol} satisfies the following properties: 
\begin{enumerate}
    \item  $v_{refk}(t)$ and $i_k(t)$ are sinusoidal with frequency $\omega_0$.
     \item  $K_{volk}(s)$ has poles at $\pm j\omega_0$, that is, $\dfrac{1}{s^2+\omega_0^2} $ is a factor in $K_{volk}(s)$.
     \item  The closed-loop system from reference voltage, $v_{rk}(s)$, to output voltage, $v_k(s)$, is stable.
\end{enumerate}
\end{assumption}
The following theorem shows that if all external inputs (i.e., $v_{refk}(t)$ and $i_k(t)$) to the system are sinusoidal with frequency $\omega_0$, then the inverter voltage follows the sinusoidal reference $v_{rk}(t)$ with zero steady-state error. 
\begin{theorem}\label{thm:VfollowVREFandI}
Under Assumption \ref{assmp:VfollowVREFandI}, the closed loop system in steady state satisfies, 
\begin{equation*}
      v_k(t) = v_{rk}(t) = v_{refk}(t) - R_{vk} i_k(t).
\end{equation*}
Furthermore, $G_k(j\omega_0)=1$ and $Z_k(j\omega_0)=R_{vk}$ in \eqref{eq5}.
\end{theorem}
\begin{proof}
 See Appendix.
\end{proof}
%%%%%%%%%%%%%%%%%%%%%%%%%%%%%%%%%%%%%%%%%%%%%%%%%%%%%%%%%%%%%%%%
%%%%%%%%%%%%%%%%%%%Voltage-Power Droop Controller%%%%%%%%%%%%%%%
%%%%%%%%%%%%%%%%%%%%%%%%%%%%%%%%%%%%%%%%%%%%%%%%%%%%%%%%%%%%%%%%
%  and
% isochronous $\sin (\omega_0 t)$ generation from a common clock.
\subsection{Outer Voltage-Active Power Droop}
The sinusoidal voltage reference signal, $V_{refk}(s)$, used by the inner multi-loop controllers is generated by the proposed VP-D control law. The proposed droop law does not facilitate any explicit control on reactive power ($Q$) unlike full droop control methodologies where active and reactive power are controlled separately by controlling voltage and frequency respectively. The VP-D law for the $k^{th}$ VSI is 
\begin{equation}\label{eq:Voltage_droop}
    E_k = E^{*} - n_k (P_k - P^{*}_k),
\end{equation} where $E_k$, $E^{*}$, $n_k$, $P_k$, and $P^{*}_k$ are inverter output voltage magnitude, reference output voltage magnitude, droop coefficient, measured output active power, and reference active power of the $k^{th}$ inverter. Here, $v_{refk}(t) = E_k \sin (\omega_0t)$. To control $E_k$, it is necessary to measure $P_k$. This is usually achieved by measuring instantaneous power (i.e., $p_k = v_k(t)i_k(t)$) and passing it through a first order low pass filter (LPF) with bandwidth $\omega_P$ which is designed to be much smaller than the voltage control loop. Reference \cite{salapaka_viability_2014-2} proposes a framework for obtaining a near-exact solution for determining the active power for a single VSI serving a load.
% This architecture facilitates the isochronous voltage reference generation and maintains a system frequency, $\omega_0$, throughout mitigating the inherent issues that arise when implementing $Q\sim \omega$ droop law as described below. 
%%%%%%%%%%%%%%%%%%%%%%%%%%%%%%%%%%%%%%%%%%%%%%%%%%%%%%%%%%%%%%%%
%%%%%%%%%%%%Common Clocked $\sin{(\omega_0 t)}$ Generation%%%%%%
%%%%%%%%%%%%%%%%%%%%%%%%%%%%%%%%%%%%%%%%%%%%%%%%%%%%%%%%%%%%%%%%
\subsection{Isochronous Architecture }
An isochronous architecture can be achieved by providing a common timing signal to all inverters in the network. The advantages of such an architecture over frequency droop are delineated earlier in this article.
% : i) the ability to operate the microgrid at a single frequency, thereby eliminating deviation of frequency that arises when using frequency-droop control; ii) simplified control of the inverters by implementing only VP-D control without explicitly controlling reactive power and, iii) enabling a `plug and play' nature of the microgrid. 
GPS is used to obtain absolute time, time interval and frequency with precision up to a nanosecond (ns). GPS receiver modules have the ability to measure high precision timing signals which are readily available and inexpensive. This makes them an ideal candidate for time synchronization applications. It is assumed here that there is accurate and precise synchronization of the real-time clocks of the networked distributed devices to the common clock and that the clock synchronization protocol is in compliance with IEEE 1588-2008 standards \cite{4579760}.
\par It is imperative that such an architecture be robust to any momentary loss of access to the common clock by one or more inverters in the system. Such a protocol is assumed and its mitigation is not a point of discussion for this article. However, mitigation of such scenarios can be achieved by employing a corrective mechanism, for example by employing a local clock at each VSI for satisfactory operation. This is possible since a local clock can provide a reasonable estimate of GPS time signal for tens of seconds, eventually shifting the operating point in response to load changes causing degradation in power system performance. For conditions of continued signal loss, the VSI unit cannot participate in system wide load sharing and should be isolated from the rest of the system. Methodologies such as radio reference signal could be used if there is sufficient concern about system wide loss of GPS reference signal \cite{majumder2009improvement}. Similarly, \cite{bellini_robust_2015} utilizes three different PLL structures to guarantee robustness versus delay or loss in communication to synchronize the phase and the frequency of all of the devices connected to the microgrid.
%%%%%%%%%%%%%%%%%%%%%%%%%%%%%%%%%%%%%%%%%%%%%%%%%%%%%%%%%%%%%%%%
%%%%%%%%%%%%%%%%%%%%%%%%%%%%%%%%%%%%%%%%%%%%%%%%%%%%%%%%%%%%%%%%
%%%%%Analysis of Isochronous Operation of Multi-inverter System%
%%%%%%%%%%%%%%%%%%%%%%%%%%%%%%%%%%%%%%%%%%%%%%%%%%%%%%%%%%%%%%%%
\section{Multi-Inverter System Isochronous Operation}\label{analysis}
% In this section, a detailed analysis of isochronous operation of multi-inverter system is conducted by analyzing the performance of individual inverters, the nature of circulating currents and bounds on voltage regulation at PCC. The network of Fig. \ref{multiinverter} is considered to be purely resistive (also aided by incorporating a virtual resistance term) and that the active power is appreciably affected by the magnitude of the voltage. Also, situations like mismatch (denoted as $\Delta$) in total commanded power and actual total load demand as well as non-zero values of $P_k-P_k^*$ for $k^{th}$ inverter are taken into consideration as these are unavoidable in practical systems.  
\subsection{Performance of Individual Inverters}
In this section, the performance of the multi-inverter system with the proposed VP-D approach is analyzed. % Applying VP-D control to a system of multiple inverters enables a flexible, reliable, and resilient power system with well-regulated voltage and frequency.
A benefit of the VP-D approach is that reactive power sharing among individual inverters is maintained within a bound without having control of reactive power explicitly defined as a control objective. Theorem \ref{thm:noQflowRload} provides analytical expressions for output voltage and line currents of networked VSI units.
%%%%%%%%%%%%%%%%%%%%%%%%%%%%%%%%%%%%%%%%%%%%%%%%%%%%%%%%%%%%%%%%
%%%%%%%%%%%%%%    Theorem 2     %%%%%%%%%%%%%%%%%%%%%%%%%%%%%%%%
%%%%%%%%%%%%%%%%%%%%%%%%%%%%%%%%%%%%%%%%%%%%%%%%%%%%%%%%%%%%%%%%
\begin{theorem}\label{thm:noQflowRload}
If the conditions of Assumption \ref{assmp:VfollowVREFandI} hold, then for any complex load $Z_L$ for the network shown in Fig. \ref{fig:multiinverter}, the steady-state voltage and current at the output of the $k^{th}$ inverter are, respectively, $v_k(t)=V_k\sin({\omega_0t+\psi_k})$ and $i_k(t)=I_k\sin({\omega_0t+\phi_k})$, where,
\begin{align*}
I_k = \frac{\gamma_k}{r_k+R_{vk}}\sum_{m=1}^N \frac{E_m}{r_m+R_{vm}}, &\hspace*{2mm}
\tan\phi_k = \frac{\sin\angle\alpha}{\cos\angle\alpha - \frac{\beta_k}{|\alpha|}}\\
V_k = \sqrt{(E_k - R_{vk}I_k)^2 + (R_{vk}I_k)^2}, &\hspace*{1mm} \tan\psi_k = \frac{\sin\phi_k}{\cos\phi_k - \frac{E_k}{R_{vk}I_k}}
\end{align*}
where $\gamma_k$, $\beta_k$, and $\alpha$ are defined in the proof.
\end{theorem}
\begin{proof}
See Appendix.
\end{proof}
\color{black}
\begin{corollary}\label{cor2.1}
In the case of purely resistive load ,i.e., $Z_L=R_L$, the steady state voltage and current at the output of the $k^{th}$ inverter are
$v_k(t)=V_k\sin\omega_0t$ and
$i_k(t)=I_k\sin\omega_0t$ for $V_k, I_k \in {\rm I\!R}_+$. In other words, phase lags $\phi_k = \psi_k = 0$.
\end{corollary}
\begin{proof}
From the expression of $\alpha$ in Theorem \ref{thm:noQflowRload}, it can be clearly stated that if $Z_L=R_L$, then $\alpha$ is a pure real number and thus $\angle\alpha=0$. Therefore,\\
$\tan \phi_k = 0 \implies \phi_k = 0 \implies \tan \psi_k = 0 \implies \psi_k = 0$
\end{proof}
 \begin{remark}\label{rem2.2}
\normalfont With the assumed control architecture and without any explicit control of reactive power, there will be no extraneous reactive power flows if the load is purely resistive. This also holds true when there is an active power mismatch. 
% In other words, there will not be any extraneous reactive power flow if the load, $Z_L$, is purely resistive, irrespective of the presence of active power mismatch. 
This is validated and elaborated in Section \ref{hardwareresult} and Fig. \ref{hardshare}(a).
\end{remark}
% % \begin{comment}reactive power demand ($Q_L$) is equally shared among all inverters i.e. $Q_1 = Q_2 = \ldots = Q_N = Q$ where $NQ = Q_L$.

% \end{corollary}
% \begin{proof}
% \color{red} Non completed. Need more clarification\color{black}
% \end{proof}
% \end{comment}
% \end{corollary}
%%%%%%%%%%%%%%%%%%%%%%%%%%%%%%%%%%%%%%%%%%%%%%%%%%%%%%%%%%%%%%%%
%%%%%%%%%%%%%%    Theorem 3     %%%%%%%%%%%%%%%%%%%%%%%%%%%%%%%%
%%%%%%%%%%%%%%%%%%%%%%%%%%%%%%%%%%%%%%%%%%%%%%%%%%%%%%%%%%%%%%%%
\subsection{Design Guideline for Droop Coefficients}
In the following theorem the phase difference between inverter currents $i_k$ and $i_j$ is characterized in the general setting for complex linear loads.
 \begin{theorem}\label{thm:diffInPhaseOfCurrents}
 Let $\delta_k=\frac{E_k-E^*}{E^*},$ $\bdelta=\left(\begin{array}{lll}\delta_1 & \ldots &\delta_N\end{array}\right)^T,$ $\blam_v:=(\begin{array}{lll}\frac{1}{r_1+Z_v}&\ldots&\frac{1}{r_N+Z_v}\end{array})^T,$ $\bxi=\frac{\blam_v}{\blam_v\bone},$ $\nu=\dfrac{Z_L\blam_v^T\bone}{1+Z_L\blam_v^T\bone}.$ Suppose the conditions of Assumption \ref{assmp:VfollowVREFandI} hold, then the phase $\phi_k$ of $i_k$ differs from the phase $\phi_j$ of $i_j$, given as, $\tan(\phi_k-\phi_j)=$
 \begin{equation}\label{eq:deltaPhase}\displaystyle
 \frac{Img(\nu)\dfrac{\delta_k-\delta_j}{ 1+\bxi^T\bdelta}}{\left(Re(\nu)-\dfrac{1+\delta_k}{(1+\bxi^T\bdelta)}\right)\left(Re(\nu)-\dfrac{1+\delta_j}{(1+\bxi^T\bdelta)}\right)+Img^2(\nu)}.
 \end{equation}
 \end{theorem}
\begin{proof}
See Appendix.
\end{proof}
% Theorem \ref{thm:diffInPhaseOfCurrents} shows that for any two inverters $j,k$, the phase difference between inverter output currents $i_j$ and $i_k$ is related to the tightness of the output voltage regulation with regard to the reference voltage. It also provides practical guidelines for designing the droop control.
% \begin{theorem}\label{thm:diffInPhaseOfCurrents}
% Define
% \begin{align*}
% \delta_k=\dfrac{E_k-E^*}{E^*}, \hspace*{5mm} \bdelta=\left(\begin{array}{lll}\delta_1 & \ldots &\delta_N\end{array}\right)^T\\
% \blam_v:=(\begin{array}{lll}(r_1+R_{v1})^{-1}&\ldots&(r_N+R_{vN})^{-1}\end{array})^T 
% \blam_v:=(\begin{array}{lll}\dfrac{1}{r_1+R_{v1}}&\ldots&\dfrac{1}{r_N+R_{vN}}\end{array})^T
% \end{align*}
% If the conditions of Assumption \ref{assmp:VfollowVREFandI} hold, $\delta_k=O(\epsilon)$, and $|Z_L\blam^T\bone|\gg 1$, then the phase $\phi_k$ of $i_k$ differs from the phase $\phi_j$ of $i_j$, as follows:
% \vspace{-0.3cm}
% \begin{align*}
% \tan(\phi_k-\phi_j) &\approx \left(\sum_{m=1}^N \dfrac{|Z_L|\sin\theta_L}{r_m+R_{vm}}\right) (\delta_k-\delta_j)
% \end{align*}
% %where $\bxi$, $\nu$ are defined in the proof.  
% \end{theorem}
% \begin{proof}
% See Appendix.
% \end{proof} 
\begin{remark}
\normalfont For purely resistive load, $\tan (\phi_k-\phi_j)=0$ as $\nu$ becomes real. This also follows from Corollary \ref{cor2.1}.
\end{remark}
\begin{corollary}\label{cor3.1}
Suppose $\delta_k=O(\epsilon)$ and $|Z_L\blam^T\bone|\gg 1$
then,
\begin{equation*}
\tan(\phi_k-\phi_j)\approx \left(\sum_{m=1}^N \dfrac{|Z_L|\sin\theta_L}{r_m+R_{vm}}\right) (\delta_k-\delta_j)
\end{equation*}
\end{corollary}
\begin{proof}
Assuming that $\delta_k$ are small so that:
${1+\delta_k}/{(1+\bxi^T\bdelta)}\approx 1 $, it follows that  
\[\begin{array}{lll}
\tan(\phi_k-\phi_j)&\approx&-\blam_v^T\bone|Z_L|^2\operatorname{Im}({1}/{ Z_L})(\delta_k-\delta_j)\\
&=&(\delta_k-\delta_j)\displaystyle\sum_{m=1}^N {|Z_L|\sin\theta_L}/({r_m+R_{vm}}),
\end{array}\]
where $\cos \theta_L$ is the power factor of the load $Z_L$.
\end{proof}
% \begin{remark}[Theorem \ref{thm:diffInPhaseOfCurrents}]
% \normalfont Assuming that the voltage amplitude $E_k$  is close to $E^*$, resulting in small $\delta_k$, it is expected that $\delta_k-\delta_j$ is small. Typically, the interconnect resistance $r_k$  and the virtual resistance $R_{vk}$ are both small and thus for typical cases $|Z_L\blam_v^T\bone|\gg 1$ should also hold. Thus the assumption in Corollary~\ref{cor3.1} are not restrictive.
% \end{remark}
% \begin{remark}[Theorem \ref{thm:diffInPhaseOfCurrents}]
% \normalfont The relationship in Corollary \ref{cor3.1} shows that to ensure  the phase difference of the inverter currents to be small, the values $E_k$ have to be close to $E^*.$ Indeed, for the phase difference of the currents to be smaller than $\epsilon\ll 1$, it follows that  $|\delta_k-\delta_j|\leq \epsilon \frac{\sum_{m=1}^N \frac{1}{r_m+R_{vm}}}{|Z_L|\sin \theta_L}$ needs to be satisfied.  
% \end{remark}
\begin{remark}
\normalfont Corollary \ref{cor3.1} extends the relationship in (\ref{eq:deltaPhase}) to show that to ensure  the phase difference in the inverter currents to be small, the values $E_k$ have to be close to $E^*.$ Indeed, for the phase difference of the currents to be smaller than $\epsilon\ll 1$, it follows that  $|\delta_k-\delta_j|\leq \epsilon \frac{\sum_m \frac{1}{r_m+Z_v}}{|Z_L|\sin \theta_L}$ needs to be satisfied, where $\cos \theta_L$ is the p.f. of the load, $Z_L$. Moreover, when compared to $\vert Z_L \vert$ the branch resistance, $r_k$, and virtual resistance, $R_v$, are both designed small. Thus, the assumptions in Corollary \ref{cor3.1} are not restrictive.
\end{remark}
\begin{remark}
\normalfont The expression for $\delta_k$ depends only on $E_k$, whereas all other parameters are load and network parameters. Moreover, the deviation of $E_k$ from $E^*$ is dictated by the outer droop law. Thus, this analysis aids the appropriate choice of $n_k$ and the allowable mismatch from the reference active power $P_k^*$ that is to be delivered, E.g., selecting droop gains $n_i, n_j$ for any two inverters $i, j$, such that $n_i(P_i^*-P_i) = n_j(P_j^*-P_j)$ will keep reactive power flows in check.
\end{remark}

\subsection{Stability Analysis}
%%%%%%%%%%%%%%%%%%%%%%%%%%%%%%%%%%
%%%%%%%%%%%%%%%%%%%%%%%%%%%%%%%%%%%%%%%%%%%%%%%%%%%%%%%%%%
One of the important aspects to be investigated for the multi-inverter system operation is its stability. The non-linear multi-inverter closed-loop system can be represented as,
\begin{align}\label{eq:nonLinear}
    \dot{X}(t) = f (X(t),u(t))
\end{align}
where $X(t)$ denotes states of the system including both physical variables and internal control variables as shown in Table \ref{tb:kthVSI} and the input to the system, $u(t)$ is $v_{PCC}(t)$. The objective is to develop a small signal model of the multi-inverter system by linearizing the system (\ref{eq:nonLinear}) around a stable operating point.  We consider the $d_{k}-q_k$ axis as the local reference frame for the $k^{th}$ VSI operating at a rotating frequency $\omega_k$.  The individual VSI state equations are derived in terms of their individual local reference frame. The local reference frames are transformed to the global reference (D-Q) frame which is chosen to be the PCC reference frame operating at $\omega_{com}$. This allows for generalizing the VP-D architecture beyond single-phase microgrids. We first consider the state equations of individual VSIs in their local reference frames represented by $\Theta_k := \theta_{com} - \theta_k$, where $\theta_{com}$ represents the synchronous reference frame angle at PCC and $\theta_k$ represents the deviation of the $k^{th}$ VSI from the synchronous reference frame (Fig. \ref{fig:multiinverter} b). The translation between local  $d_{k}-q_k$ frame to global D-Q frame is given by, 
\[T_{\Theta_k} = \begin{bmatrix} \cos\Theta_k  &  -\sin\Theta_k \\
\sin\Theta_k  & \cos\Theta_k \end{bmatrix},\]
We consider, 
\begin{align*}
    v_{kdq} &= [ v_{kd} \;\, v_{kq}]^T, \,  i_{kdq} = [i_{kd} \;\,i_{kq}]^T, i_{Lkdq} = [i_{Lkd} \;\, i_{Lkq} ]^T,
\end{align*}
where, $\begin{bmatrix} x_{kd} \\ x_{kq} \end{bmatrix} = T_{\Theta_k} \begin{bmatrix} x_\alpha \\ x_\beta \end{bmatrix} $. $x_\alpha $ and $x_\beta$ are the $\alpha - \beta$ components of a given state $x$. For the single phase case, $x_\beta =0$. Note that having  access  to  the  common  clock signal  facilitates  generation of $\alpha-\beta$ components  without  added complexities  of  a PLL  implementation. This is given by the following transformation matrix, 
% \begin{figure*}[b]
% \noindent\makebox[\linewidth]{\rule{0.8\paperwidth}{0.4pt}}
% \begin{align}\label{eq:A_kmat}
% A_{sys} =
% \begin{bmatrix}-\omega_p&0&0&0&0&0&0&I_{kd}&I_{kq}&V_{kd}&V_{kq}\\-n_k\omega_{cv}k_{iv}&-\omega_{cv}&0&0&0&0&0&-k_{iv}\omega_{cv}&0&-R_{vk}&0\\0&0&-\omega_{cv}&0&0&0&0&0&-k_{iv}\omega_{cv}&0&-R_{vk}\\0&k_{ic}\omega_{ci}&0&-\omega_{ci}&0&-k_{ic}\omega_{ci}&0&0&0&1&0\\0&0&k_{ic}\omega_{ci}&0&-\omega_{ci}&0&-k_{ic}\omega_{ci}&0&0&0&1\\0&0&0&\frac1{L_{fk}}&0&\frac{-r_{fk}}{L_{fk}}&\omega_0&0&0&0&0\\0&0&0&0&\frac1{L_{fk}}&-\omega_0&\frac{- r_{fk}}{L_{fk}}&0&0&0&0\\0&0&0&0&0&\frac1{C_{fk}}&0&0&\omega_0&\frac{-1}{C_{fk}}&0\\0&0&0&0&0&0&\frac1{C_{fk}}&-\omega_0&0&0&\frac{-1}{C_{fk}}\\0&0&0&0&0&0&0&\frac1{L_k}&0&\frac{-r_k}{L_k}&\omega_0\\0&0&0&0&0&0&0&0&\frac1{L_k}&-\omega_0&\frac{-r_k}{L_k}\end{bmatrix} + \sum\limits_{k=1} ^ N\frac{-1}{L_k}T_{\Theta_k}^{-1}T_{\Theta_j}
% \end{align}
% \end{figure*}
    
    % ***************************************
\noindent Note that in the case of perfect synchronization with the GPS clock signal, $\theta_k = 0$ for all VSI as a result of which all  inverters  can  be  represented  in  the  reference frame eliminating the need for translating each VSI's states into the  common  reference  frame. Each of these VSI units comprise an outermost droop control, inner multi-loop voltage and current controllers, output LC-filter connected to the PCC through line impedances and are terminated at the common load. Combining these dynamics, the small signal state space model of the $k^{th}$ single inverter system is given as,
\begin{align}
    \Delta \dot{X}_k = A_k\; \Delta X_k +B_k\; T_{\Theta_k} ^{-1}\Delta  V_{PCCDQ}
\end{align}
\noindent where $\Delta X_k = 
    [\Delta{P_k}\;\, \Delta \phi_{kdq} \;\,\Delta\gamma_{kdq}\;\,\Delta i_{Lkdq} \;\, \Delta v_{kdq} \;\, \Delta i_{kdq} ]^T$, 
    $B_k = \begin{bmatrix}0_{9\times2}\\ \frac{-1}{L_k} \mathbf{I}_{2\times2}\end{bmatrix}$ and  $A_k$ is given as ($\Gamma_{kv}=-R_{vk}K_{iv}\omega_{cv}$),\\
    % \begin{figure}[h]
    % \centering
    %     \includegraphics[scale=0.45,trim={1.1cm, 0.4cm, 0cm, 0cm},clip]{Images/A_mat.png}
    %     \label{fig:my_label}
    % \end{figure}
\begin{table}[b]
\centering
\caption{States of the multi-inverter system}
\label{tb:kthVSI}
\begin{tabular}{|cc|}
\hline
\textbf{Parameters} & \textbf{States}     \\ \hline 
Droop \& LPF        & $P_k$                 \\
Voltage Controller  & $\phi_{k,dq}$       \\
Current Controller  & $\gamma_{k,dq}$     \\
LC filter           & $i_{Lk,dq},v_{k,dq}$ \\
Line parameters     & $i_{k,dq}$          \\
Load, $Z_L$         & $i_{Load,DQ}$       \\ \hline
\end{tabular}
\end{table}
% \begin{figure}
%     \centering
%     \includegraphics[scale=0.33]{Images/frames2.png}
%     \caption{Representation of VSI $j$ and VSI $k$ with respect to common reference frame.}
%     \label{fig:frames}
% \end{figure}
Combined model of the N inverter system can now be written as \cite{stabilitygreen,stabilitykimball}, 
\begin{align}\label{eq:ss_Ninv}
    \Delta \dot{X} = \mathbf{A} \;\Delta X + \mathbf{B}\; \Delta V_{PCC,DQ},
\end{align}
where
    $\Delta X =  [\Delta X_1 \,\; \Delta X_2 \,\;\dots \,\; \Delta X_N]^T$, 
    $\mathbf{A} = \text{diag}(A_1, \; A_2, \; \dots \; ,A_N)$ and $\mathbf{B} =  \begin{bmatrix} B_1\; T_{\Theta_1}^{-1} & B_2\; T_{\Theta_2} ^{-1} & \dots & B_N\; T_{\Theta_N} ^{-1} \end{bmatrix}^T$.\\
% \begin{align}\label{eq:loadeqn}
%     \Delta \dot{i}_{Load,DQ} = -\dfrac{R_{L}}{L_{L}} \Delta i_{Load,DQ} +\dfrac{1}{L_{L}} \Delta V_{PCC,DQ}
% \end{align}
% where, $\Delta i_{LoadDQ} = T_{\theta}\sum_{k=1} ^N \Delta i_{kdq}$. Thus, 
% \begin{align}\label{eq:completess}
% &\begin{bmatrix}  \Delta \dot{X} \\ \Delta \dot{i}_{Load}  \end{bmatrix} = \begin{bmatrix}
% \mathbf{A} & 0 \\ \mathbf{0} & \dfrac{R_L}{L_L} \end{bmatrix}\begin{bmatrix} \Delta X \\ \Delta i_{Load}\end{bmatrix} + \begin{bmatrix}
% \mathbf{B}T_{\theta}^{-1} \\ \dfrac{1}{L_L}
% \end{bmatrix} \Delta V_{PCC}
% \end{align}
Considering the dynamics of the lumped complex ($Z_L$) load and considering only a resistive load for simplicity, we obtain:\vspace{-0.2cm}
\begin{align}\label{eq:R-load}
 \Delta V_{PCC,DQ} &= \sum_{k=1}^N T_{\Theta_k} \Delta i_{kdq} R_L
\end{align}
Thus, using (\ref{eq:R-load}) in (\ref{eq:ss_Ninv}) we get, $\Delta \dot{X} = \underbrace{(\mathbf{A} \; + R_L\mathbf{B}\; \mathbf{T_{\Theta}})}_{\mathbf{A}_{sys}} \Delta X,$
where, $\mathbf{T_{\Theta}}$ is a matrix that captures the coupling of VSI interconnection at PCC. Stability analysis in the presence of a complex $R-L$ load can be obtained in a similar manner. 
% Complete development of the linearized multi-inverter system will be presented in the final revision.
%   where the block matrix of $\mathbf{A}_{sys}$ for the $k^{th}$ VSI is given as
%   \begin{align*}
% \left[\begin{array}{ccccccc}
%   \ddots &\multicolumn{5}{c}{\vdots} & \iddots\\
%   &\multicolumn{5}{c}{\tilde{A}_k} &  \\
%   \hdots & 0_{1\times7}&\frac1{L_k}&0&-\frac{(r_k+R_L)}{L_k}&\omega_0 & \hdots \\
%     & \multicolumn{2}{c}{0_{1\times8}}&\frac1{L_k}&-\omega_0&-\frac{(r_k+R_L)}{L_k}&  \\
%   \iddots& \multicolumn{5}{c}{\vdots}& \ddots
% \end{array}\right]
% \end{align*}
% $\tilde{A}_k$ is the state matrix for $k^{th}$ VSI associated with states $\tilde{X}_k =  [\Delta{P_k}\;\, \Delta \phi_{kdq} \;\,\Delta\gamma_{kdq}\;\,\Delta i_{Lkdq}\;\,\Delta v_{kdq} ]^T$.\\
The stability analysis of the system is given by the eigenvalues of $\mathbf{A}_{sys}$. For a two inverter setup the eigenvalues of the system are shown in Fig. \ref{fig:sensitivity_Analysis} setup configured with the parameters of Table \ref{tb:param_power}. The results are obtained for the case of an asymmetric power sharing ratio of 2.2. The line paramerters are chosen to be $\text{Line}_1 = 0.2~ \Omega + j\omega_o~0.1~mH$  and $\text{Line}_2 =  0.2~ \Omega + j\omega_o~0.15~mH $.
%%%%%%%%%%%%%%Two seprate Figures for Eigen Analysis %%%%%%%%%%%%%
% \begin{figure}
%     \centering
%     \subfloat[]{\includegraphics[scale = 0.15, trim={0cm, 0cm, 0cm, 0cm},clip]{Images/stabilitiy_R.eps}}\\\vspace{-0.35cm}
%     \subfloat[]{\includegraphics[scale = 0.15]{Images/stabilitiy_RL.eps}}
%     \caption{Eigenvalues of the two inverter network for stability analysis with (a) resistive load only, (b) R-L Load with p.f. 0.95.}
%     \label{fig:sensitivity_Analysis}
% \end{figure}
%%%%%%%%%%%%%%%%%%%%SINGLE FIG for EigenAnalysis%%%%%%%%%%%%%%%%%%%
\begin{figure}
    \centering
    \includegraphics[scale=0.17]{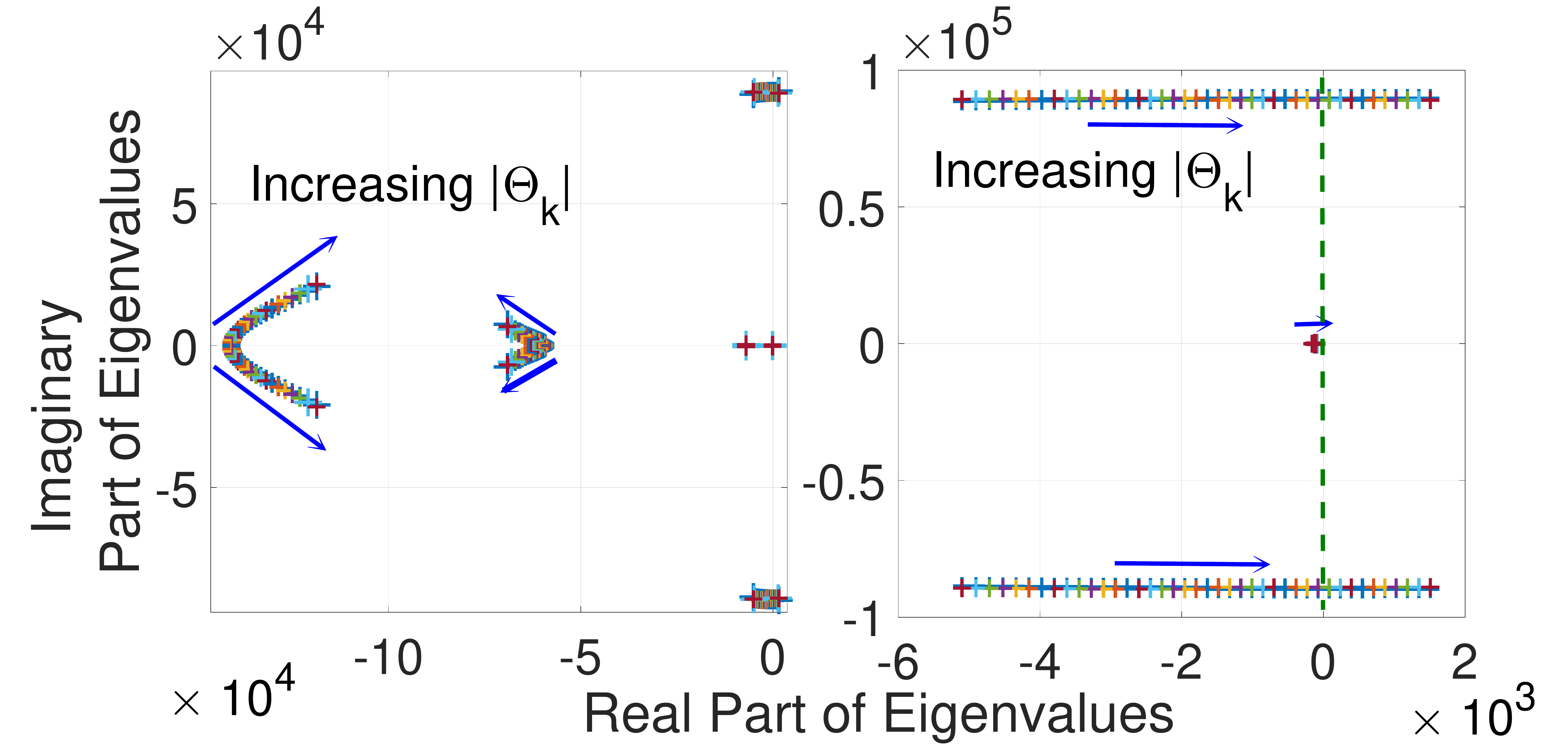}
    \caption{Eigenvalues of the two inverter network for stability analysis  with $\delta_k$ changing from $0^{\circ}$ to $\pm 10^{\circ}$ for a resistive load.}
    \label{fig:sensitivity_Analysis}
\end{figure}
\begin{remark}
    The small signal stability analysis reveals that the stability of the system is maintained in the VP-D architecture when deviation, $\theta_k$, for all individual VSI units with respect to the common reference frame are $\le \pm 5^{\circ}$ . The presence of larger inductive loads (low power factor) results in driving the eigenvalues of the linearized system towards the origin.
\end{remark}
% \begin{align*}
%     \begin{bmatrix}
%     \dot{P_k} \\ \dot{\phi_k} \\ \dot{\gamma_k} \\ \dot{\psi_k} \\ \dot{\Gamma_k} \\ \dot{i_{Lk}} \\ \dot{v_k} \\ \dot{i_k} \\
%     \end{bmatrix}
%     = \begin{bmatrix}
%     - \omega_p  P_k(t)+ \omega_p v_k(t)i_k(t) \\ \gamma_k \\ (E_k^*-n_k(P_k-P_k^*))\sin\omega_0t-R_{vk}i_k ...\\ ... - v_k-\omega_0^2\phi_k \\ \Gamma_k \\ -\omega_0^2\psi_k +i_k + k_{vp}((E_k^*-n_k(P_k-P_k^*))....\\...\sin\omega_0t -R_{vk}i_k-v_k)+k_{vi}\phi_k - i_{Lk} \\ -\dfrac{R_f}{L_f}i_{Lk} +\dfrac{1}{L_f}(\Ccancel[red]{v_k} +k_{cp}(i_{refk - i_{Lk}})-k_{ci} \psi_k)...\\... - \Ccancel[red]{\dfrac{1}{L_f}v_k(t)} \\ -\dfrac{1}{C_f}i_k(t)+\dfrac{1}{C_f}i_{Lk}(t)\\-\dfrac{r_k}{L_k}i_k(t) +\dfrac{1}{L_k}v_k(t) -\dfrac{1}{L_k}v_{PCC}(t)
%     \end{bmatrix}
% \end{align*}

\vspace{-1cm}
%%%%%%%%%%%%%%%%%%%%%%%%%%%%%%%%%%%%%%%%%%%%%%%%%%%%%%%%%%%%%%%%
%%%%%%%%%%%%%%%%%%%%%%%%%%%%%%%%%%%%%%%%%%%%%%%%%%%%%%%%%%%%%%%%
%%%%%%%%%%%%%%%%%%%%%Results%%%%%%%%%%%%%%%%%%%%%%%%%%%%%%%%%%
%%%%%%%%%%%%%%%%%%%%%%%%%%%%%%%%%%%%%%%%%%%%%%%%%%%%%%%%%%%%%%%%
%%%%%%%%%%%%%%%%%%%%%%%%%%%%%%%%%%%%%%%%%%%%%%%%%%%%%%%%%%%%%%%%
%%%%%%%%%%%%%%%%%%%%%%%%%%%%%%%%%%%%%%%%%%%%%%%%%%%%%%%%%%%%%%%%
%%%%%%%%%%%%%%%%%%%%%%%%%%%%%%%%%%%%%%%%%%%%%%%%%%%%%%%%%%%%%%%%
%%%%%%%%%%%%%%%%%%%%%%%%%%%%%%%%%%%%%%%%%%%%%%%%%%%%%%%%%%%%%%%%
\section{Results}\label{results}
%%%%%%%%%%%%%%%%%%%%%%%%%%%%%%%%%%%%%%%%%%%%%%%%%%%%%%%%%%%%%%%%
%%%%%%%%%%%%%%%%%%%%%%Simulation Results%%%%%%%%%%%%%%%%%%%%%%%
%%%%%%%%%%%%%%%%%%%%%%%%%%%%%%%%%%%%%%%%%%%%%%%%%%%%%%%%%%%%%%%%
\begin{figure}[t]
\centering
{\includegraphics[scale=0.16,trim={0cm 0cm 0cm 0cm},clip]{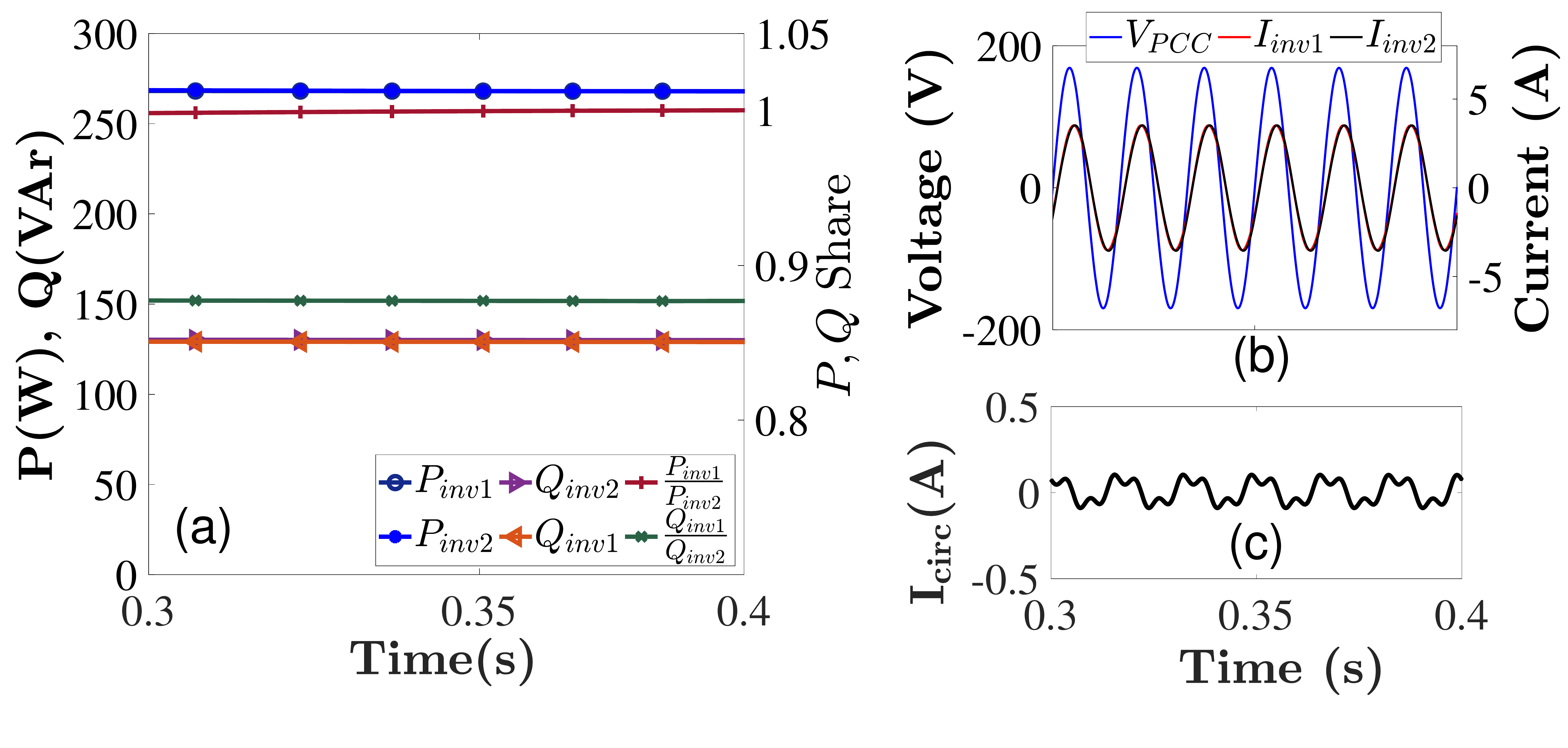}}
% \vspace{-0.09cm}
% \subfloat[]{\includegraphics[scale=0.16,trim={1.8cm 7cm 18cm 1cm},clip]{soft_vi1i2.pdf}}
% ~~
% \subfloat[]{\includegraphics[scale=0.16,trim={0.5cm 7cm 25cm 1cm},clip]{soft_circ.pdf}}
\caption{Simulation results of (a) $P$-$Q$ share, (b) voltage at the PCC and inverter output currents, and (c) circulating current between two parallel inverters.}\label{fig:vi1i2icirc}
\end{figure}

\begin{table}[b]
\centering
\caption{Parameters for Experimental Prototype}
\centering
\label{tb:param_power}
\begin{tabular}{|c|c|}
\hline
\textbf{Parameters}                      & \textbf{Values}                             \\ \hline
Cut off frequency ($\omega_p$) & $1 ~(Hz)$
\\ \hline
Nominal Voltage at PCC           & 120 \textit{V}                                        \\ \hline
Nominal Frequency at PCC ($\omega_0$)    & 60  \textit{Hz}                                        \\ \hline
Branch Resistance            & 0.2 $\Omega$                                        \\ \hline
Inverter DC Link Voltage              & 250  \textit{V}                                       \\ \hline
Inverter Filter Inductance           & 0.063  \textit{mH}                                        \\ \hline
Inverter Filter Capacitance     & 1  $\mu F$                                          \\ \hline
Rated Power of $\text{Inverter}_1$ & 0.6   kVA                                      \\ \hline
Rated Power of $\text{Inverter}_2$  & 0.6    kVA                                     \\ \hline
\end{tabular}
\end{table}
\subsection{Simulation Results}
To validate the viability of VP-D-based control of a multi-inverter system, a MATLAB/SIMULINK-based switching model of two inverters sharing a common load at PCC in a LV resistive network is simulated. The network and inverter parameters used in the simulation are listed in Table \ref{tb:param_power}.
% A choice of smaller \(\Sigma_{i} n_iP_i^*\) results smaller phase difference between inverter current waveforms as per Theorem \ref{thm:diffInPhaseOfCurrents}. Moreover, identical droop gains for inverters further reduces the phase difference between inverter output currents. A detailed design consideration of droop characteristics is given in \cite{salapaka_viability_2014-2}.
Fig. \ref{fig:vi1i2icirc} shows the voltage and current output results from a simulation where two parallel inverters are serving a common load of 0.58 kVA with $p.f. = 0.9$ (lagging). Both inverters are operated with $n_1\mathrm{P_{inv_1}^*}=n_2\mathrm{P_{inv_2}^*=2.5}$ and $\mathrm{P_{inv_1}^*=P_{inv_2}^*=250 W}$. Fig. \ref{fig:vi1i2icirc}(a) shows that the $P$ share between two inverters, i.e., $\mathrm{P_{inv_1}}/\mathrm{P_{inv_2}}$, is maintained at 1.07. Note here that without having a separate $Q$ control strategy, the reactive power sharing between the inverters is nearly even. Here, $\mathrm{Q_{inv_1}}/\mathrm{Q_{inv_2}}$ is 0.88. The mismatch between actual measured and reference output active power for both the inverters, primarily because of finite line losses, results in a small voltage deviation from $E^*$ (around -0.5$\%$) at the PCC. The voltage waveform at the PCC and the output current are shown in Fig. \ref{fig:vi1i2icirc}(b). Fig. \ref{fig:vi1i2icirc}(c) shows the circulating current between the two inverters, which is significantly low compared to the rated output current.
%%%%%%%%%%%%%%%%%%%%%%%%%%%%%%%%%%%%%%%%%%%%%%%%%%%%%%%%%%%%%%%%
%%%%%%%%%%%%%%%%%%%%%% Comparison with Full droop %%%%%%%%%%%%%%%%%%%%%%%%%%%%%%%%%%%%%%%%%%%%%%%%%%%%%%%%%%%%%%%%%%%%%%%%%%%%%%%%%%%%%%%%%%%%%%%%%%%%%%%%%%%%%%%%%%%%%%%%%%%%%%%%%%%%%%%%%%%%%%%%%%%%%%
\subsubsection{Mismatched line parameters}\par Simulations for a number of scenarios are considered each for symmetric as well as asymmetric active power sharing between the two inverters with mismatched reference and injection values of active power for the two inverters. Total reference power commanded from each inverter is $P^*_{inv1}=P^*_{inv2} = 0.6 $ kW whereas the actual load, $Z_L$, is a series $R$-$L$ load with $R=11.52 \ \Omega$ and $L = 0.02293$ H. $R_v = 0.2 \ \Omega$ for both inverters. The obtained reactive power share is  $0.9 \le \frac{Q_{inv1}}{Q_{inv2}} \le 1.03$, when the line resistance of inverter 2,  satisfies $ 0.6~  r_{k1} \le r_{k2} \le r_{k1}$. For the cases when the line resistances are highly mismatched ( i.e., $ r_{k2} = 0.5~ r_{k1}$), $\frac{Q_{inv1}}{Q_{inv2}}$, is 0.868 with a worst-case sharing ratio of 0.735 obtained when $ r_{k2} = 0.1~ r_{k1}$. This result demonstrates the performance of the proposed architecture in worse-case conditions toward mismatches in active power references and line parameter.
%%%%%%%%%%%%%%%%%%%%%%%%%%%%%%%%%%%%%%%%%%%%%%%%%%%%%%%%%%%%%%%%%%%%%%%%%%%%%%%%%%%%%%%%%%%%%%%%%%%%%%%%%%%%%%%%%%%%%%%%%%%%%%%%%%%%%%%%%%%%%%%%%%%%%%%%%%
%%%%%%%%%%%%%%%%%%%%%% Comparison with Full droop %%%%%%%%%%%%%%%%%%%%%%%%%%%%%%%%%%%%%%%%%%%%%%%%%%%%%%%%%%%%%%%%%%%%%%%%%%%%%%%%%%%%%%%%%%%%%%%%

%%%%%%%%%%%%%%%%%%%%%%%%%%%%%%%%%%%%%%%%%%%%%%%%%%%%%%%%%%%%%%%%%%%%%%%%%%%%%%%%%%%%%%%%%%%%%%%%%%%%%%%%%%%%%%%%%%%%%%%%%%%%%%%%%%%%%%%%%%%%%%%%%%%%

%%%%%%%%%%%%%%%%%%%%%%%%%%%%%%%%%%%%%%%%%%%%%%%%%%%%%%%%%%%%%%%%
%%%%%%%%%%%%%%%%%%%%%%Plug n Play%%%%%%%%%%%%%%%%%%%%%%
%%%%%%%%%%%%%%%%%%%%%%%%%%%%%%%%%%%%%%%%%%%%%%%%%%%%%%%%%%%%%%%%
%%%%%%%%%%%%%%%%%%%%%%%%%%%%%%%%%%%%%%%%%%%%%%%%%%%%%%%%%%%%%%%%
\subsubsection{Plug-and-play capability}
A three VSI network is considered. At $t=0$ s VSI-1 is on and synchronized with the local clock signal serving a load, $Z_L = 3 \text{kW } (0.97 \text{ pf lagging})$. At $t=0.1$ s and $t=0.25$ s, VSI-2 and VSI-3 set individual clock receive flags high respectively to receive synchronized clock pulses for reference generation. The output voltages at the capacitor of VSI-2 and VSI-3 are thus generated at these instants. In order to demonstrate the plug and play capability, VSI-2 is connected to the common load and VSI-1 at $t=0.2$ s and VSI-3 is then connected to the rest of the network at $t=0.3$ s. Symmetric power sharing is desired at all VSIs. Output voltages, line currents and circulating currents are shown in Fig. \ref{fig:VIplugnplay}. Fig. \ref{fig:PQplugnplay} shows the symmetric sharing of the load, $Z_L$, with active power sharing ratio of 1 as well as reactive power sharing ratio of $> 0.95$ at steady state even when no $Q$ control is implemented. Finally, at $t=0.4$s, a load transition to 4.5 kW (0.97 p.f.lagging) occurs. 
% \begin{figure}
%     \centering
%     \includegraphics[scale=0.23]{Images/Clock_signals_paper.png}
%     \caption{GPS based clock signals to synchronize reference generation for inverters (a) GPS clock pulses received by VSI-2; (b) VSI-2 local clock pulses; (c) VSI-2 reference generation; (d) GPS clock pulses received by VSI-3; (e) VSI-3 local clock pulses; (f) VSI-3 reference generation.}
%     \label{fig:clockSignal}
% \end{figure}
\begin{figure}[t]
\centering
\includegraphics[scale=0.22,trim={0cm 1.5cm 0cm 1cm},clip]{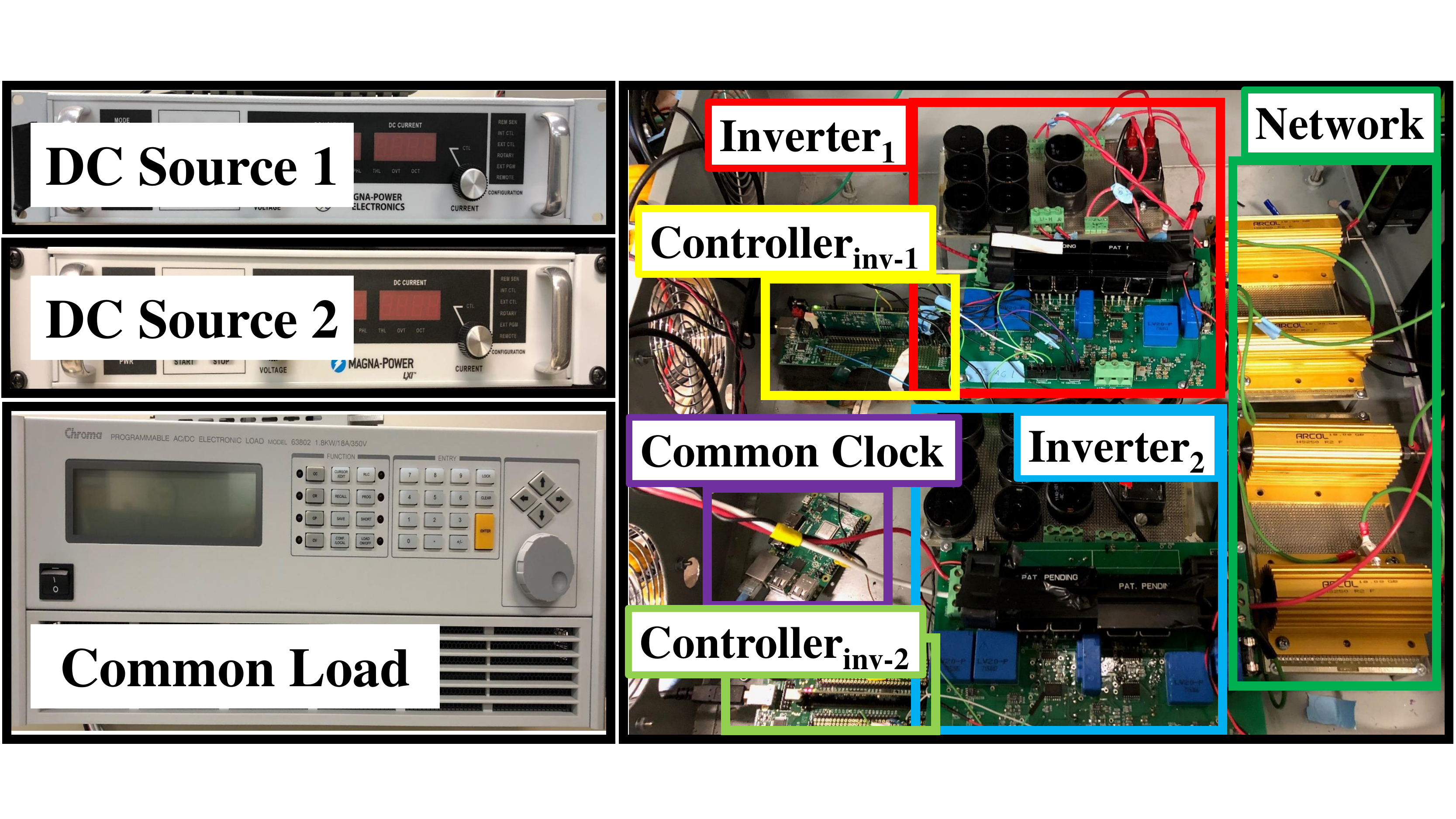}
\caption{Experimental prototype for 1.2 kVA two inverter setup with common clock and network resistances.}
\label{fig:setup}
\end{figure}
\subsubsection{Comparison with full droop ($P-V/Q-f$)}
The three VSI network as above is considered while implementing a full droop architecture with $P-V$ and $Q-f$ droop laws.  At $t=0$ s VSI-1 is on and serving a load, $Z_L = 3 \text{kW } (0.97 \text{ pf lagging})$. At $t=0.1$ s and $t=0.3$ s, VSI-2 and VSI-3 are synchronized and connected to the network. The presence of an integrator in the $Q-f$ droop implementation results in initial transients in output voltage and current of the VSIs that cause in higher transient and steady state circulating currents. This is apparent at $t=0.1$ s and $t=0.3$ s when the VSIs are turned on (see Fig. \ref{fig:fulldroopcomprVI}). At $t=0.4$s, a load transition to 4.5 kW (0.97 p.f.lagging) occurs. The full droop implementation also suffers from longer transient times to reach steady state in both active and reactive power sharing as well as poor reactive power sharing ratio (sharing ratio of 1 is desired) during load changes and steady state values due to mismatch in commanded and actual reactive power (see Fig. \ref{fig:fulldroopcomprPQ}). Meanwhile, the isochronous architecture provides improved transient and steady state performance during black starts, plug-and-play of VSIs and load changes even in mismatched conditions. High power sharing accuracy during plug-and-play and fast transient response demonstrates the superiority of the proposed architecture over full droop methods in MMGs.
\begin{figure}[t]
    \centering
    \includegraphics[scale=0.17]{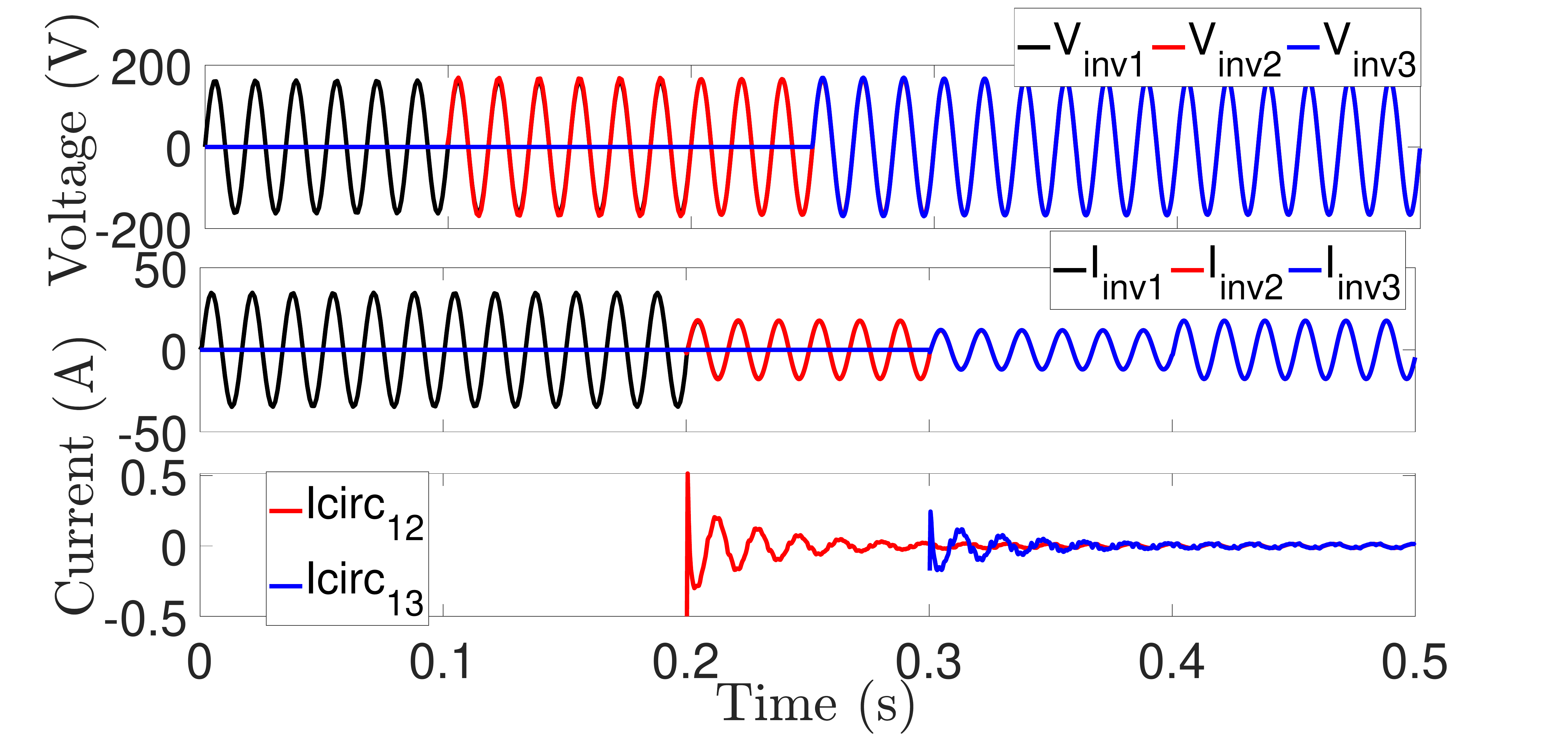}
    \caption{Output voltages, line currents and circulating current for a three inverter system with VP-D architecture for demonstrating plug and play.}
    \label{fig:VIplugnplay}
\end{figure}
\begin{figure}[t]
    \centering
    \includegraphics[scale=0.18]{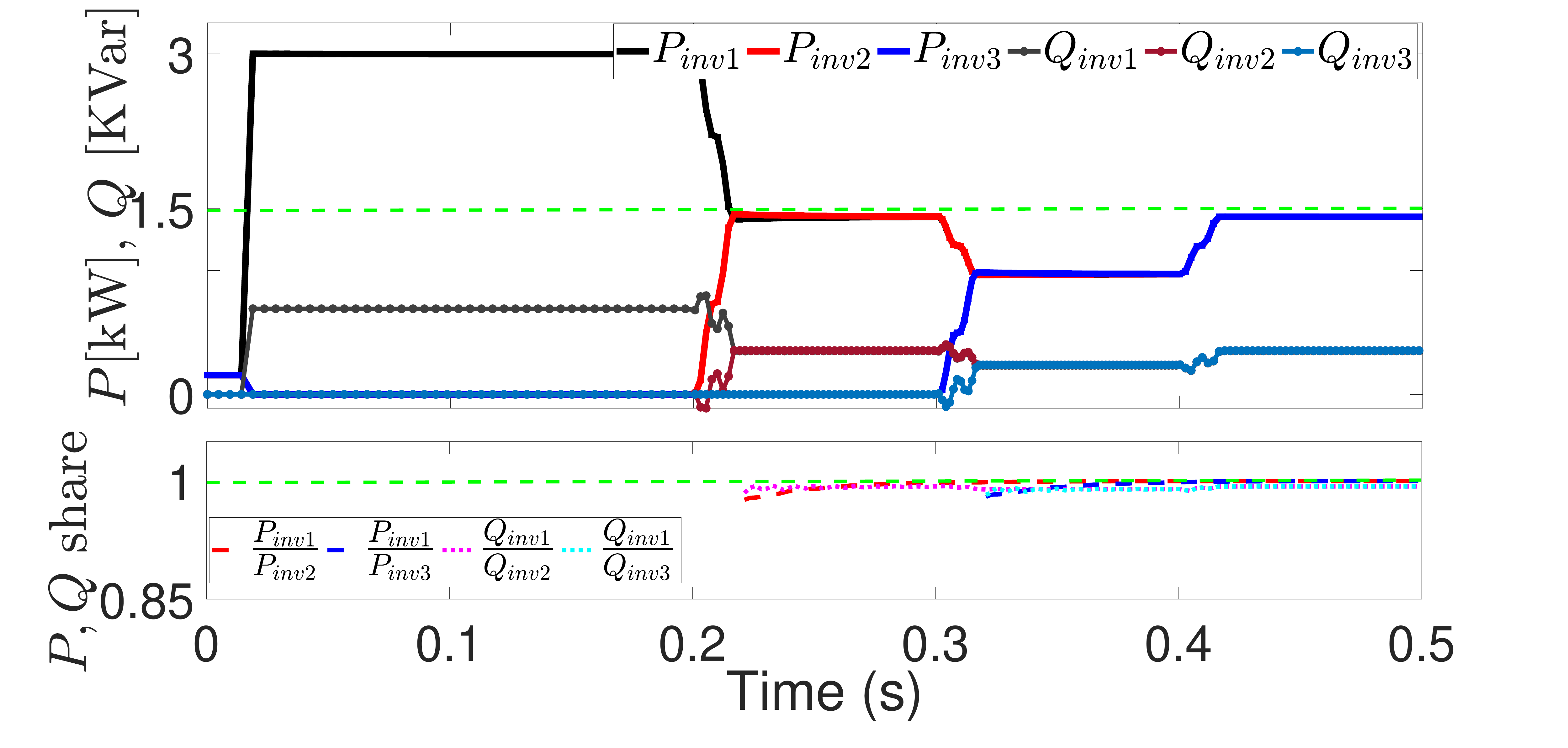}
    \caption{Active, reactive power sharing performance for the three inverter system implementing VP-D for plug-and-play.}
    \label{fig:PQplugnplay}
\end{figure}
\begin{figure}[t]
    \centering
    \includegraphics[scale=0.17]{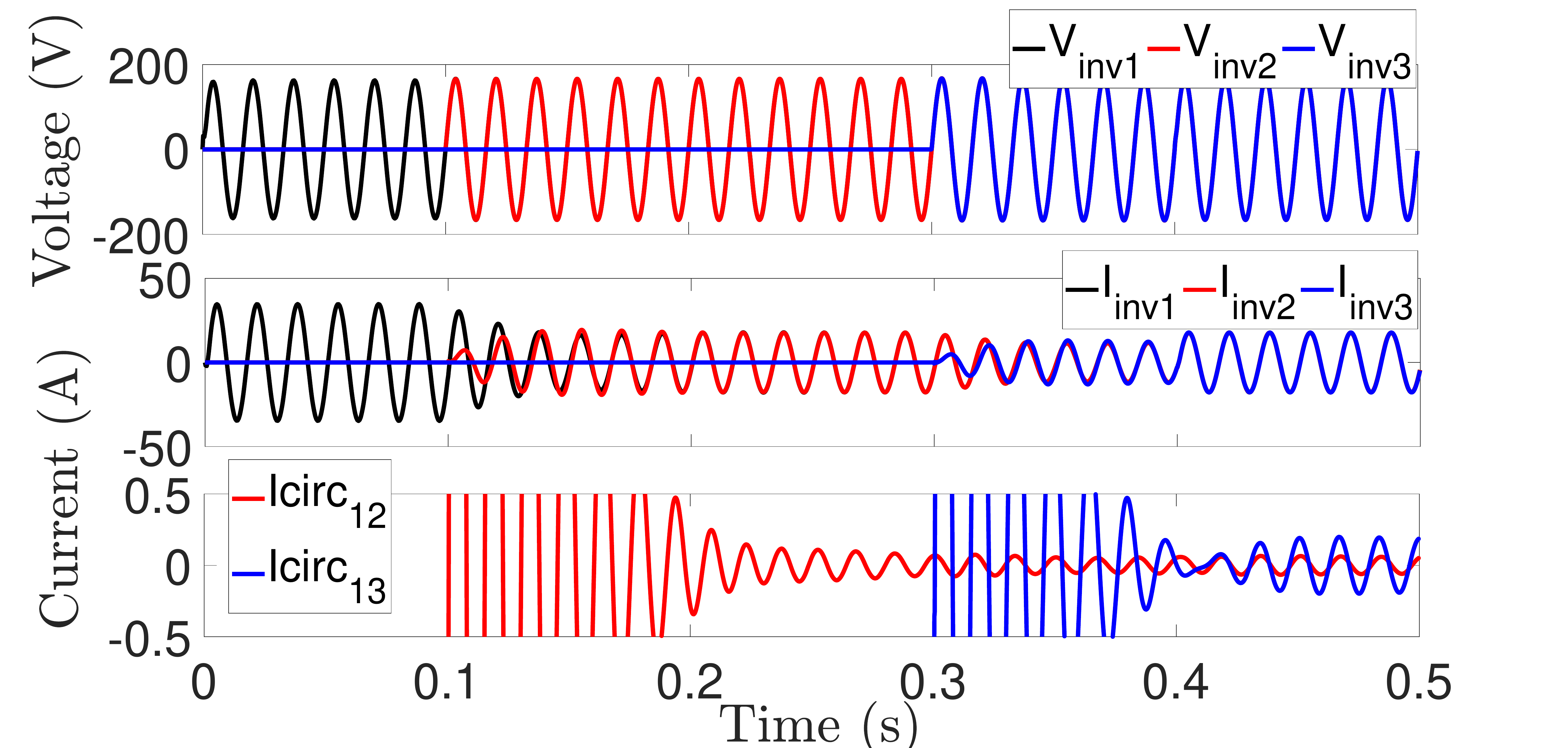}
    \caption{Output voltages, line currents and circulating current for a three inverter system implementing $P\sim V, Q \sim f$ full droop for plug-and-play.}
    \label{fig:fulldroopcomprVI}
\end{figure}
\begin{figure}[t]
    \centering
    \includegraphics[scale=0.18]{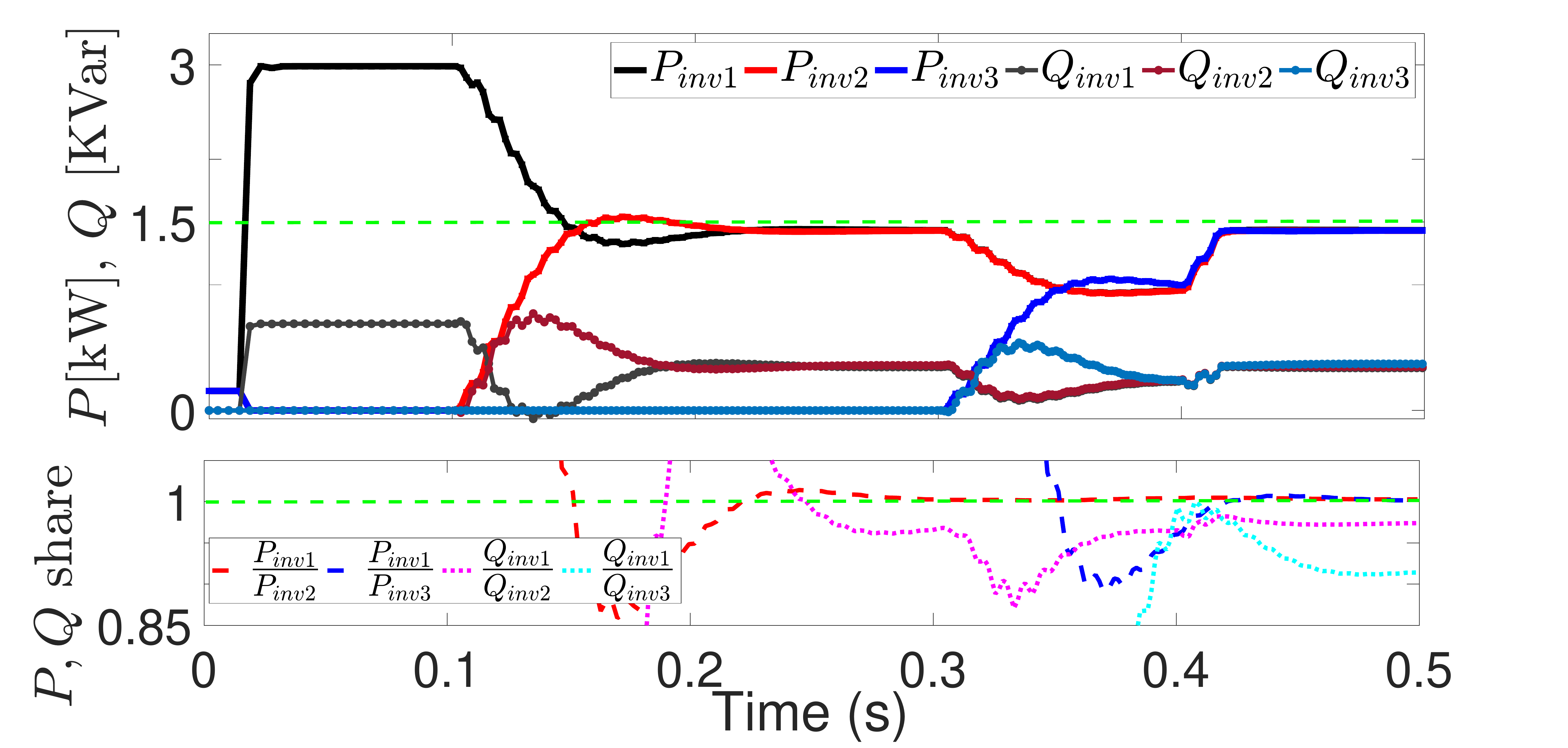}
    \caption{Active, reactive power sharing performance for the three inverter system with full droop for plug-and-play.}
    \label{fig:fulldroopcomprPQ}
\end{figure}
%%%%%%%%%%%%%%%%%%%%%%%%%%%%%%%%%%%%%%%%%%%%%%%%%%%%%%%%%%%%%%%%
%%%%%%%%%%%%%%%%%%%%%%Experimental Results%%%%%%%%%%%%%%%%%%%%%%
%%%%%%%%%%%%%%%%%%%%%%%%%%%%%%%%%%%%%%%%%%%%%%%%%%%%%%%%%%%%%%%%
%%%%%%%%%%%%%%%%%%%%%%%%%%%%%%%%%%%%%%%%%%%%%%%%%%%%%%%%%%%%%%%%
\subsection{Experimental Validation}\label{hardwareresult}
%%%%%%%%%%%%%%%%%%%%%%%%%%%%%%%%%%%%%%%%%%%%%%%%%%%%%%%%%%%%%%%%%%%%%%
%%%%%%%%%%%%%%%%%%%%%%%%%%%%%%%%%%%%%%%%%%%%%%%%%%%%%%%%%%%%%%%%%%%%%%%%%%%%%%%%%%%%%%%%%%%%%%%%%%%Controller PARAMS TABLE%%%%%%%%%%%%%%%%%%%%%%
%%%%%%%%%%%%%%%%%%%%%%%%%%%%%%%%%%%%%%%%%%%%%%%%%%%%%%%%%%%%%%%%
%%%%%%%%%%%%%%%%%%%%%%%%%%%%%%%%%%%%%%%%%%%%%%%%%%%%%%%%%%%%%%%%
% \begin{table}[!ht]
% \centering
% \caption{Controller Stage Parameters for Experimental Prototype}
% \centering
% \label{tb: param_control}
% %\renewcommand{\arraystretch}{1.3}
% \begin{tabular}{|c|c|c|}
% \hline
% {\textbf{Controller Parameters}} & {\textbf{$\text{Inverter}_1$}} & {\textbf{$\text{Inverter}_2$}} \\ \hline
% $n_k$ (V/w)                                           & 0.2                                                 & 0.2                                                 \\ \hline
% $R_{vk}$ ($\omega$)                                   & 0.2                                                 & 0.2                                                 \\ \hline
% Switching  Frequency (kHz)                            & 20                                                  & 20                                                  \\ \hline
% $E^*$ (V)                                             & 120$\sqrt{2}$                                       & 120$\sqrt{2}$                                       \\ \hline
% $\tau_P (s)$                                              & 0.1             & 0.1             \\ \hline
% \end{tabular}
% \end{table}
Experiments are carried out in a laboratory setup as shown in Fig. \ref{fig:setup}. The experimental configuration consists of two single-phase insulated-gate bipolar transistor-based H-bridge inverter units along with $LC$ output filter and controller unit for each, a purely resistive network, a variable common load served by both inverter units, and a Raspberry Pi-based (RPi) clock pulse generator. The RPi unit provides clock pulses with precise time-stamping. Each inverter consists of a digital input module to emulate a GPS receiver. All inverters synchronize their local clocks, using zero-crossing detection and pulse counting, to the RPi clock signal transmitted as a square wave. This clock signal is then used to generate sinusoidal references to individual inner-loop controllers. Each inverter unit is rated as 0.6 kVA operating and utilizing a switching frequency of 20 kHz. The active power control including active power calculation and droop control, reference voltage generation, inner voltage and current controller, and PWM generation, digital input module are implemented on a TMS320F28335 digital signal processor. The inverter and network parameters are provided in Table \ref{tb:param_power}. The value of the droop coefficients for both inverters $-$ i.e., $n_1, n_2$ are selected as $2 \times 10^{-4}$ V/W. The loop-shaping controllers $K_{curk}$ and $K_{volk}$ are designed as described in Section \ref{controller_design} with bandwidths 980 Hz and 600 Hz, respectively, for closely tracking sinusoidal references. 

% \begin{figure*}[t]
% \centering
% \subfloat{\includegraphics[scale=0.17]{Images/Sim_vi_complex_load_paper2.pdf}
%     \caption{Output voltages, line currents and circulating current for a three inverter system with VP-D architecture for demonstrating plug and play.}} ~~
% \subfloat{}

% \end{figure*}

\begin{figure*}[t]
\centering
\subfloat{\includegraphics[width=9.25cm,height = 5cm,trim={1cm 0cm 0cm 0cm},clip]{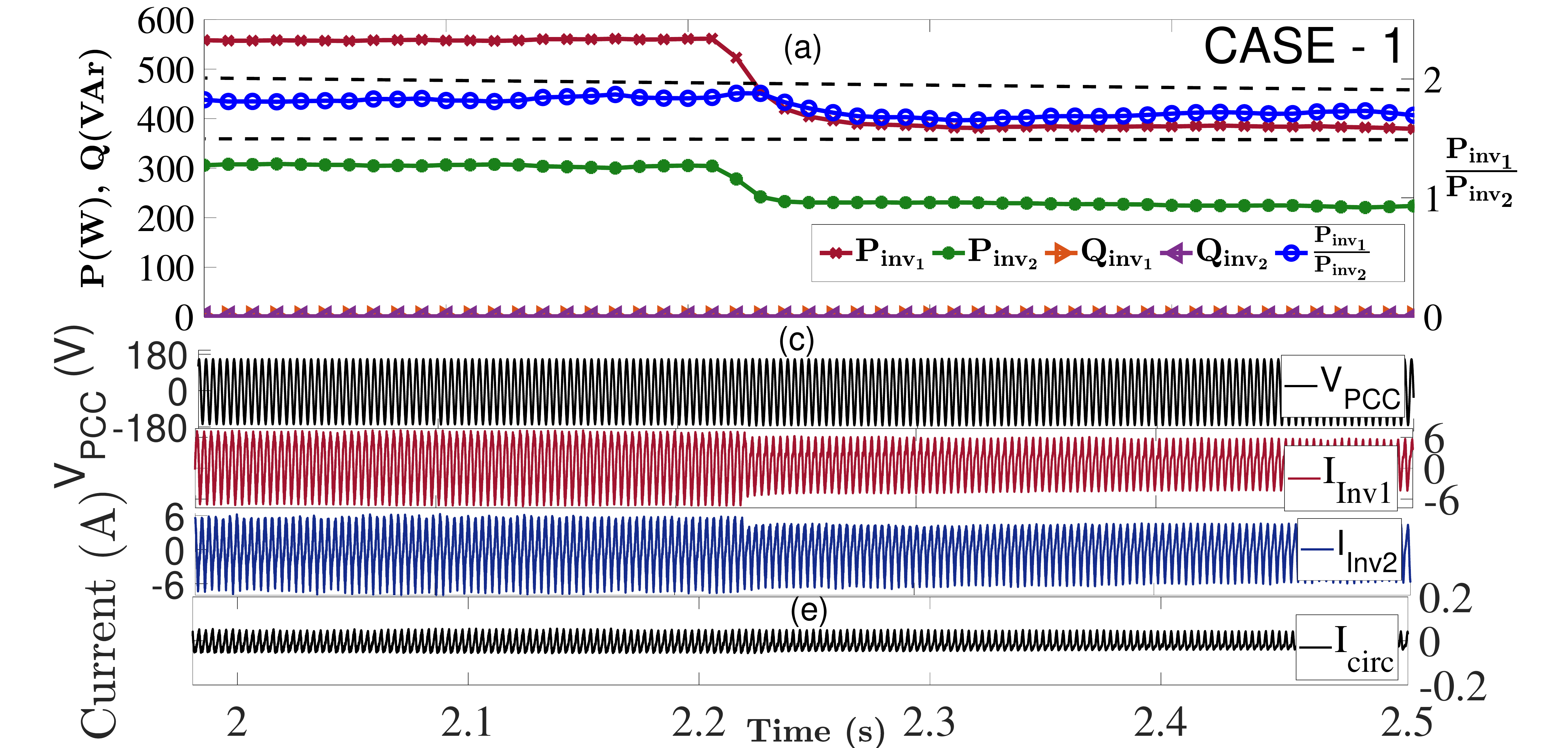}}~
\subfloat{\includegraphics[width=9.25cm,height = 5cm,trim={1cm 0cm 0cm 0cm},clip]{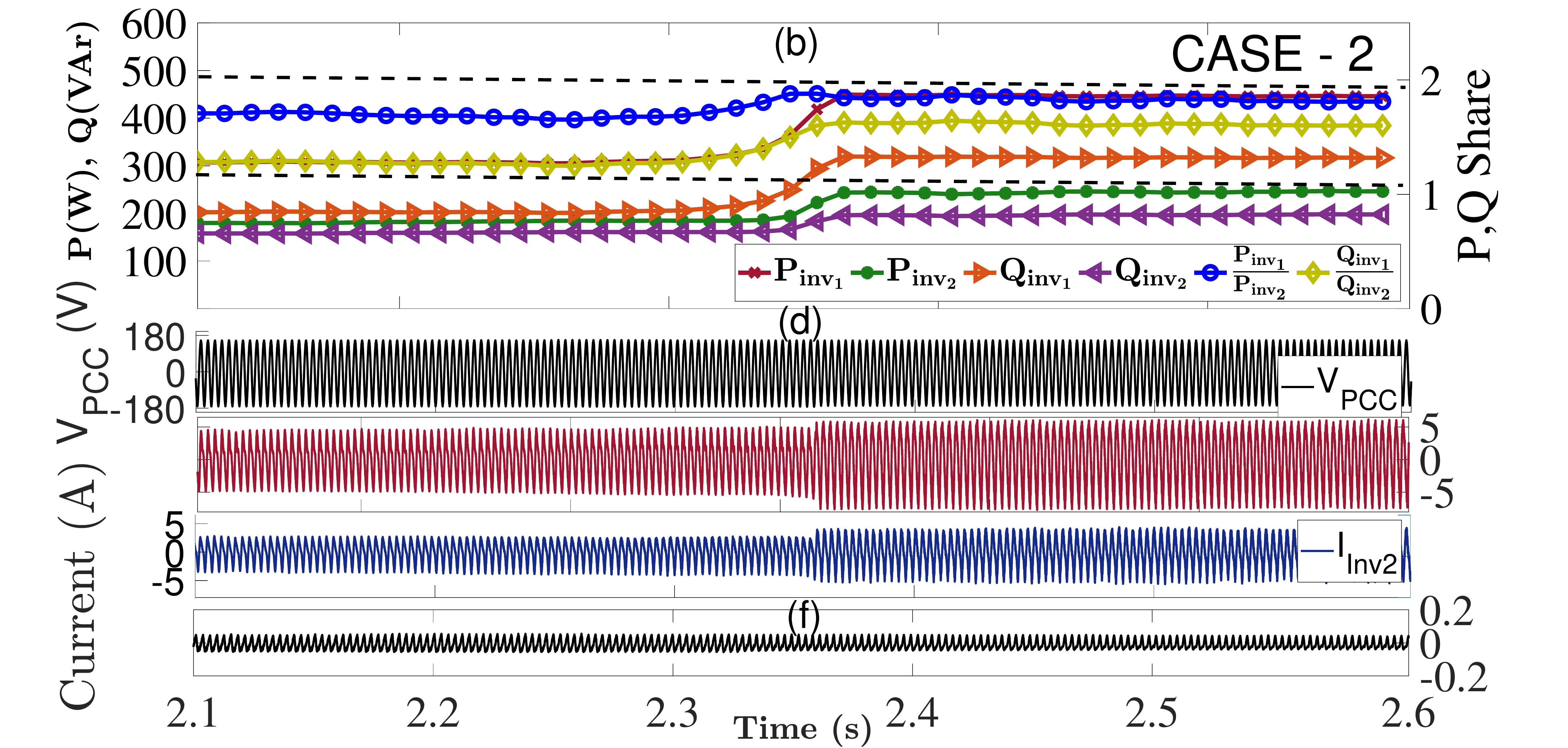}}
\caption{Hardware results of (a) $P$-$Q$ share between two parallel inverters in Case 1, (b) $P$-$Q$ share in Case 2 (c) common grid voltage and output currents in Case 1, (d) common grid voltage and output currents in Case 2 and (e) circulating current in Case 1, (f) circulating current in Case 2.Case 1 has a load of 0.87 kVA with transition at t = 2.22s to 0.61 kVA with $p.f.= 1.0$ and Case 2 has a load of 0.6 kVA with transition at t= 2.34s to 0.84 kVA with $p.f.= 0.8$.}
\label{hardshare}
\end{figure*}

The expressions of $K_{volk}$ and $K_{curk}$ are as follows:
\[\begin{array}{lll}
K_{volk}=\frac{s (s+3139) (s+1172)  (s^2 + 3328s + 8.895\times 10^6)}{1492.75(s+5390) (s+1406) (s^2 + (2\pi60)^2)}\\ \\
K_{curk} = \frac{38403 (s+1406) (s+222)}{(s+6468) (s^2 + (2\pi60)^2)}
\end{array}\]
% \begin{comment}
% % \begin{figure}[!ht]
% % \centering
% % \subfloat[]{\includegraphics[scale=0.16,trim={0.5cm 0cm 1.4cm 1cm},clip]{hard_Pshare_pf08.eps}\label{fig:hardsharepf08}}

% % \subfloat[]{\includegraphics[scale=0.164,trim={2cm 0cm 0.8cm 1cm},clip]{hard_Pshare_pf1.eps}\label{fig:hardsharepf1}}
% % \caption{Hardware results of $P$-$Q$ share between two parallel inverters serving a common load of (a) 0.6 kVA with power factor 0.8 with no transition in load demand and (b) a transition in load demand at t = 2.22s from 0.87 kVA to 0.61 kVA maintaining unity power factor}
% % \label{hardshare}
% % \end{figure}
% \end{comment}
% \begin{table}[!ht]
% \centering
% \caption{THD values for transition of Fig.\ref{fig:hardvi1i2}}
% \centering
% \label{tb:THD}
% \renewcommand{\arraystretch}{1.3}
% \begin{tabular}{|c|c|c|c|}
% \hline
% \textbf{THD}                                                          & \textbf{$\mathbf{V}_{\mathbf{Grid}}$} & \textbf{$\mathbf{I}_{\mathbf{inv_1}}$} & \textbf{$\mathbf{I}_{\mathbf{inv_2}}$} \\ \hline
% \textit{\begin{tabular}[c]{@{}c@{}}Before \\ Transition\end{tabular}} & 1.86$\%$            & 4.49$\%$            & 4.20$\%$            \\ \hline
% \textit{\begin{tabular}[c]{@{}c@{}}After\\ Transition\end{tabular}}   & 1.47$\%$            & 4.59$\%$             & 4.86$\%$            \\ \hline
% \end{tabular}
% \end{table}
% \subsubsection{Experimental Results}\label{hardwareresult}
To verify the proposed VP-D control of the multi-inverter system, two case studies are conducted for experimental validation of the active and reactive power sharing performance (as shown in Fig. \ref{hardshare}) using the hardware prototype. Case 1 considers two inverters serving a common load of $S_{tot}$ = 0.87 kVA with transition at t = 2.22 s to $S_{tot}$ = 0.61 kVA maintaining unity power factor throughout. Meanwhile, Case 2 studies two inverters serving a common load of $S_{tot}$ = 0.6 kVA with transition at t = 2.34 s to $S_{tot}$ = 0.84 kVA maintaining 0.8 power factor throughout.
\par In Case 1, both inverters are assigned with asymmetric active power references of $\mathrm{P_{inv_1}^*}$ = 0.56 kW and $\mathrm{P_{inv_2}^*}$ = 0.31 kW throughout the operation. The active power sharing ratio of $1.8 \pm 10\%$ is maintained throughout the operation, as shown in Fig. \ref{hardshare}(a). Moreover, the output reactive powers of both the inverters are zero, which is a direct validation of Corollary \ref{cor2.1}. The transient time of the inverter, which here is defined as the 90\% settling time of the inverter's output current peak after a load step, is about 0.037 s (i.e., 2.2 cycles of 60 Hz). This demonstrates the ability of the proposed controller to inject real power into the network to provide ancillary support to the microgrid in fast- (primary response) timescale regimes. Fig. \ref{hardshare}(c) shows the voltage and output current waveforms of both inverters before and after the load transition. %This validates the remarks of Theorem \ref{thm:noQflowRload}, that is even if there is no direct $Q$ control, the phase difference between inverter output current waveform and the voltage waveform is zero under purely resistive load. 
There is a finite active power mismatch ($\Delta \approx +4.44\%$) that causes a deviation ($\approx +2.5\%$) of voltage magnitude at the PCC from its nominal value. These deviations are attributed to measurement error in active power at relatively low power operating conditions. Fig. \ref{hardshare}(e) illustrates the circulating current between two inverters throughout the operation \cite{circ}. Clearly, the magnitude of the circulating current is significantly less than the nominal output current ratings of the inverters. The total harmonic distortion (THD) values of the voltage waveform at the PCC before and after the transition are 1.86\% and 1.47\%. %Similarly THD values of inverter output current waveform are tabulated in Table. \ref{tb:THD}.
These values are also maintained within acceptable standards \cite{8332112}.
Similarly, in Case 2, the inverters are assigned asymmetric active power references of $\mathrm{P_{inv_1}^*}$ = 0.3 kW and $\mathrm{P_{inv_2}^*}$ = 0.18 kW throughout the operation. The sharing ratio of active power is maintained at 1.67 throughout the operation. Moreover, reactive power is shared with a ratio of $1.2 \pm 10\%$ without implementing any dedicated $Q$ control strategies throughout the operation. Fig. \ref{hardshare}(d) illustrates the voltage and output current waveforms of both inverters before and after the load transition. The transient time of inverter current is 0.05 s (i.e., 2.99 cycles of 60 Hz). Here, $\Delta \approx -5.56\%$ that results in a deviation ($\approx -3\%$) in PCC voltage magnitude from its nominal. Fig. \ref{hardshare}(f) illustrates the circulating current between the two inverters. % 
\vspace{-0.5cm}
\subsection{Design Considerations of Multi-Inverter System}
This section provides design guidelines to restrict limits on reactive power flows for implementing VP-D control in a multi-inverter system.    
\begin{figure}[t]
\centering
\subfloat[]{\includegraphics[scale=0.12,trim={3cm 0cm 18cm 1cm},clip]{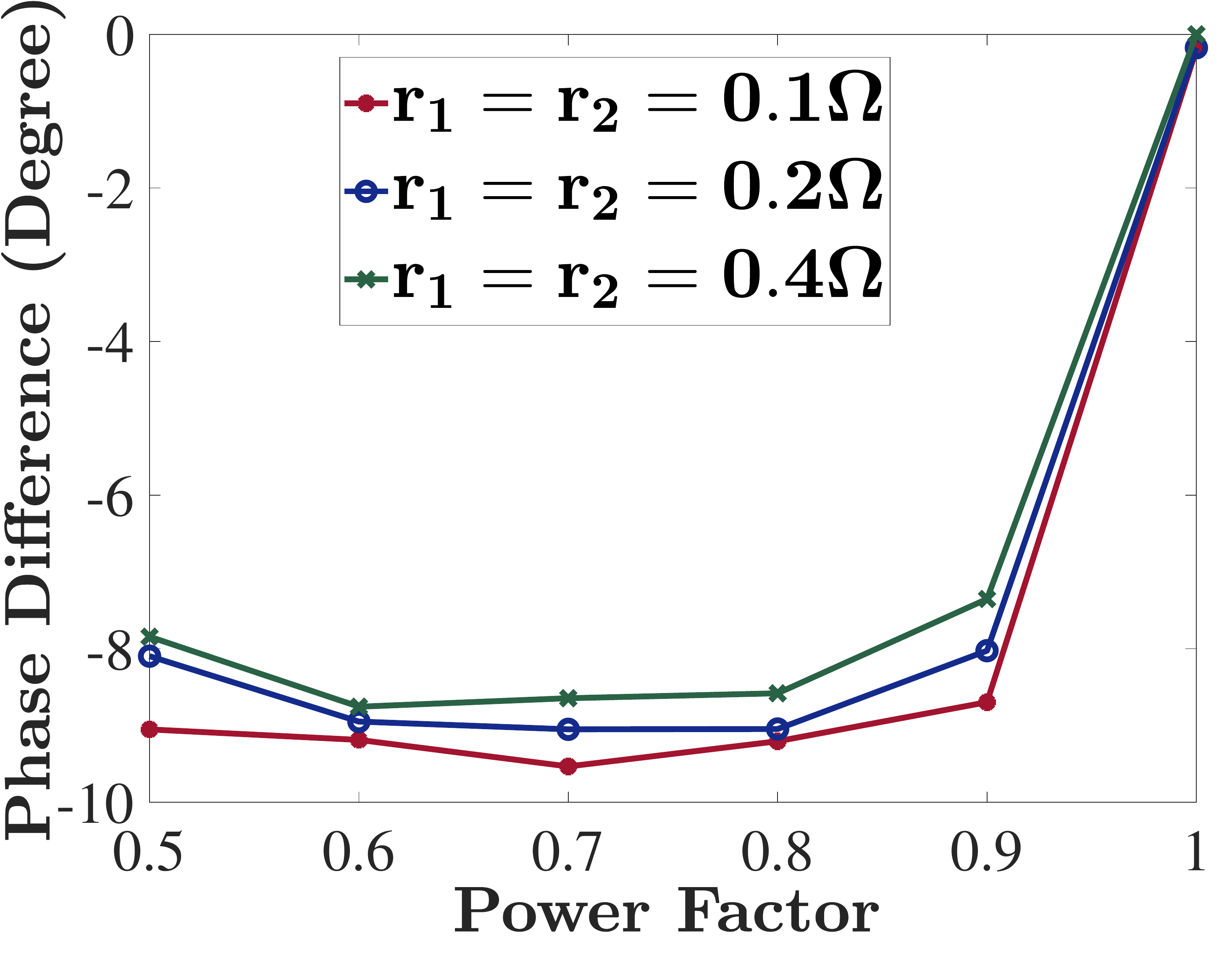}\label{fig:rkvariation}}
~~
\subfloat[]{\includegraphics[scale=0.12,trim={2.8cm 0cm 18cm 1cm},clip]{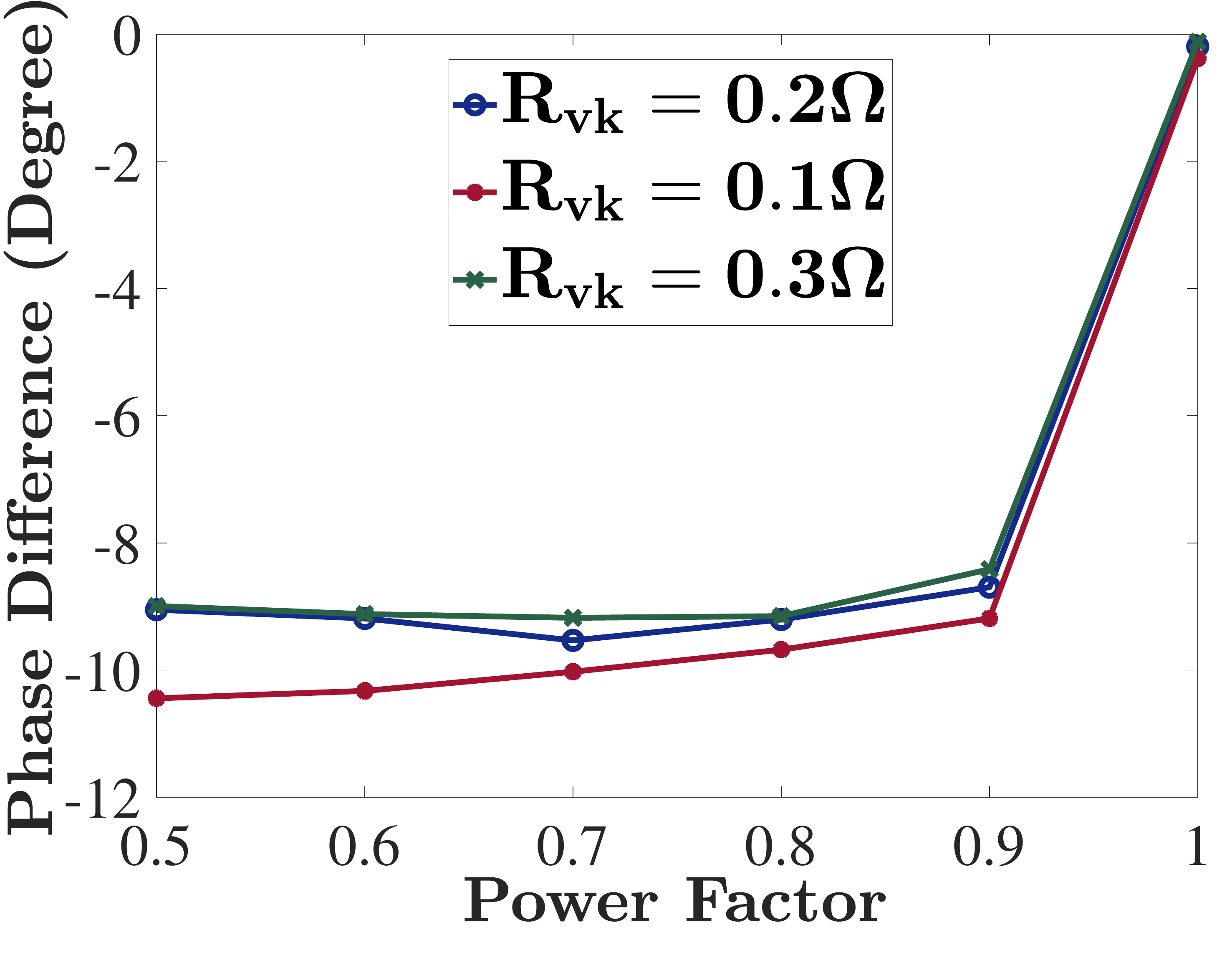}\label{fig:rvvariation}}\\
~~
\subfloat[]{\includegraphics[scale=0.12,trim={3cm 0cm 18cm 1cm},clip]{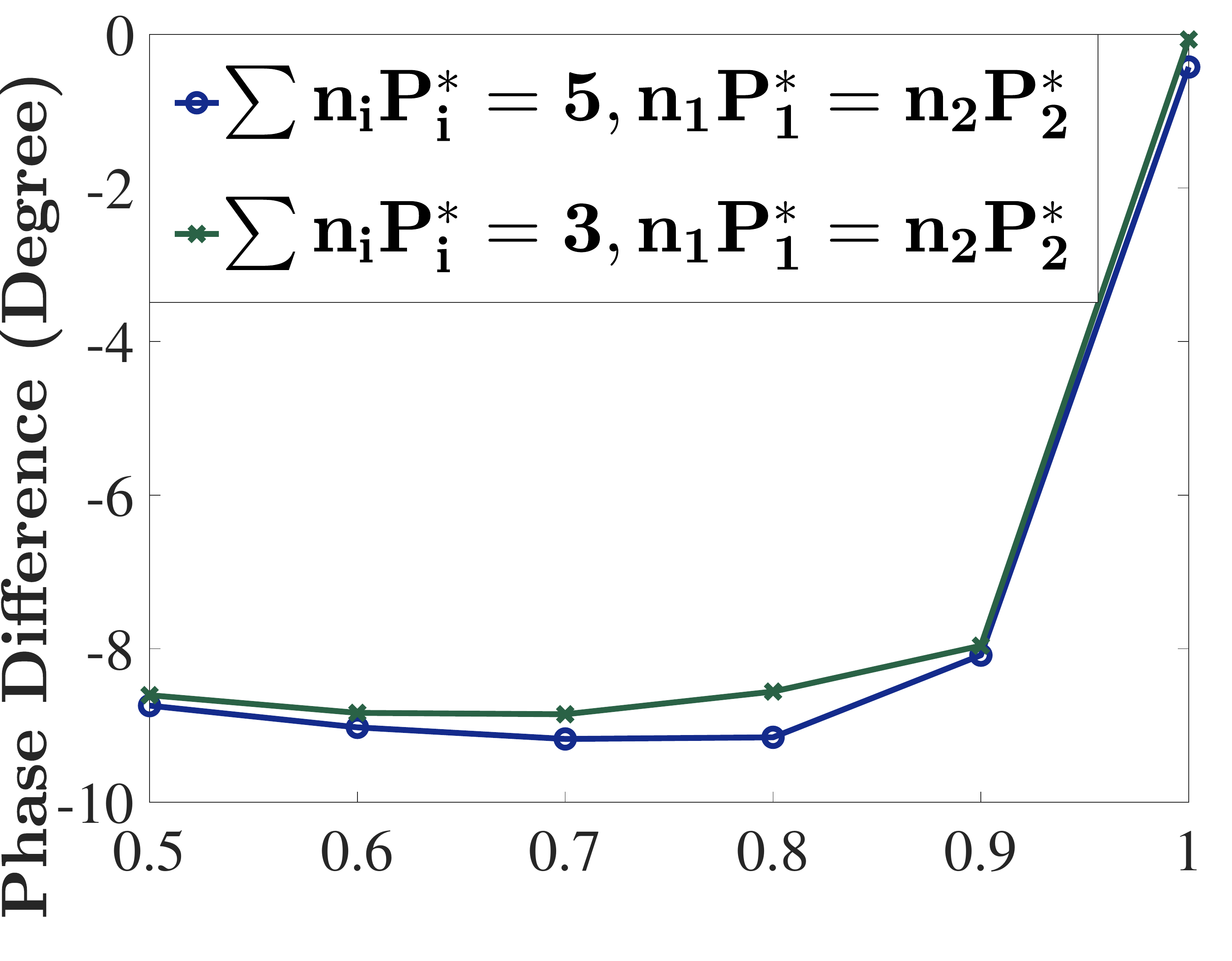}\label{fig:nipivariation}}
~~
\subfloat[]{\includegraphics[scale=0.12,trim={3cm 0cm 18cm 1cm},clip]{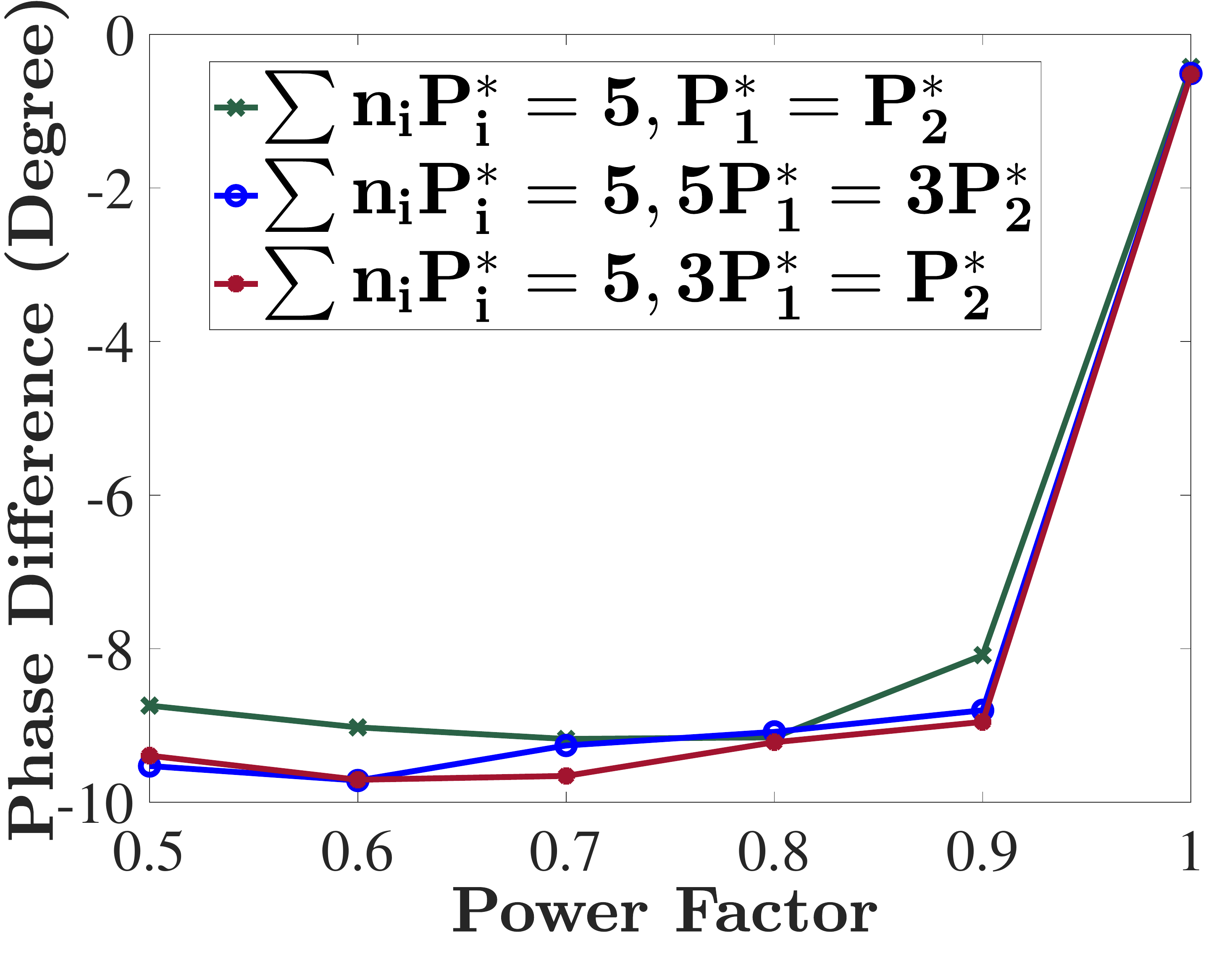}\label{fig:nipiuneqPvariation}}
\caption{Experimental results of phase difference between output currents of two inverters serving common load with varied power factors in the scenario of having various values of (a) $\Sigma_i n_iP_i^*$ with $n_1P_1^* = n_2P_2^*$, (b) $\Sigma_i n_iP_i^*$ with $n_1P_1^* \ne n_2P_2^*$.}\label{fig6}
\end{figure}
Fig. \ref{fig6} shows experimental results for the variation in phase difference (inverter currents) when an electronic load ( Fig. \ref{fig:setup}), $Z_L$ with apparent Power $S = 0.6$ kVA is varied in power factor from $p.f.=0.5$ to unity. It is important to emphasize here that even with such a low power factor of the load to be served and no explicit control over reactive power flows of the system, the phase difference between currents is maintained low ($\approx 10^\circ$ or less) by the proposed controller architecture.\\
Fig. \ref{fig6}(a) and Fig. \ref{fig6}(b) show that higher values of branch resistance and virtual resistance lead to smaller phase differences. This, however, may lead to a decrease in the accuracy of active power sharing in the network.
% A direct consequence of Theorem \ref{thm:diffInPhaseOfCurrents} is that deviation of $E_k$ from $E^{*}$ should be kept small for all inverters in order to control phase differences. This can be supplemented by choosing $n_iP_i^{*}=n_jP_j^{*}$ while designing the droop law and applying the conditions of Theorem \ref{thm:noBrnchResNoInvRes}. This implies \(E_i - E^{*} = n_i(P_i^{*}-P_i)= n_j(P_j^{*}-P_j)=E_j- E^{*}\). However, branch resistances and virtual resistances are seldom zero but are nevertheless close to small values. This allows for a choice of droop gain $n_i$ and $n_j$ according to Corollary \ref{cor4.1} to be effective. \\
From Fig. \ref{fig6}(c) and Fig. \ref{fig6}(d) it can be observed  that a choice of smaller \( \sum_{i} n_iP_i^*\) results in a smaller phase difference. These observations serve as design considerations to limit reactive power flows and implement an isochronous VP-D architecture in an LV multi-inverter system.
% \vspace{-0.5cm}
%%%%%%%%%%%%%%%%%%%%%%%%%%%%%%%%%%%%%%%%%%%%%%%%%%%%%%%%%%%%%%%%
%%%%%%%%%%%%%%%%%%%%%%%%%%%%%%%%%%%%%%%%%%%%%%%%%%%%%%%%%%%%%%%%
%%%%%%%%%%%%%%%%%%%%%%%%%%%%%%%%%%%%%%%%%%%%%%%%%%%%%%%%%%%%%%%%
%%%%%%%%%%%%%%%%%%%%%%%%%%%%%%%%%%%%%%%%%%%%%%%%%%%%%%%%%%%%%%%%
%%%%%%%%%%%%%%%%%%%%%%%%%%%%%%%%%%%%%%%%%%%%%%%%%%%%%%%%%%%%%%%%
\section{Conclusion}\label{conclusion}
In this paper, a voltage-active power half-droop was proposed and demonstrated to show high power sharing accuracy during plug-and-play and fast transient response of the proposed architecture over full-droop methods in a resistive LV ac microgrid. A loop shaping-based control law was designed for a single inverter unit to enable close regulation of the reference sinusoidal voltage signal to be tracked. A novel isochronous architecture was also proposed to maintain the system frequency and enable active power sharing in the multi-inverter system. The main advantage of the proposed architecture is to keep the reactive power flows small among the inverters without implementing an explicit $Q-f$ droop law, thereby, reducing the overall complexity of the inverter controllers. Small signal based stability analysis was conducted to show the feasibility of the proposed architecture for various scenarios of the multi-inverter system. Moreover, analytical results were derived to provide design guidelines for selecting droop gains when there are active power mismatches and uncertainties in the load present in the LV microgrid. Experimental results were provided for various scenarios that validate the performance and capabilities of the proposed controller architecture. 
\vspace{-0.5cm}
%%%%%%%%%%%%%%%%%%%%%%%%%%%%%%%%%%%%%%%%%%%%%%%%%%%%%%%%%%%%%%%%
%%%%%%%%%%%%%%%%%%%%%%%%%%%%%%%%%%%%%%%%%%%%%%%%%%%%%%%%%%%%%%%%
%%%%%%%%%%%%%%%%%%%%%%%%%%%%%%%%%%%%%%%%%%%%%%%%%%%%%%%%%%%%%%%%
%%%%%%%%%%%%%%%%%%%%%%%%%%%%%%%%%%%%%%%%%%%%%%%%%%%%%%%%%%%%%%%%
%%%%%%%%%%%%%%%%%%%%%%%%%%%%%%%%%%%%%%%%%%%%%%%%%%%%%%%%%%%%%%%%
%%%%%%%%%%%%%%%%%%%%%%%%%%%%%%%%%%%%%%%%%%%%%%%%%%%%%%%%%%%%%%%%
%%%%%%%%%%%%%%%%%%%%%%%%%%%%%%%%%%%%%%%%%%%%%%%%%%%%%%%%%%%%%%%%
%%%%%%%%%%%%%%%%%%%%%%%%%%%%%%%%%%%%%%%%%%%%%%%%%%%%%%%%%%%%%%%%
%%%%%%%%%%%%%%%%%%%%%%%%%%%%%%%%%%%%%%%%%%%%%%%%%%%%%%%%%%%%%%%%
\section*{Appendix}
%%%%%%%%%%%%%%%%%%%%%%%%%%%%%%%%%%%%%%%%%%%%%%%%%%%%%%%%%%%%%%%%
%%%%%%%%%%%%%%%%% Proof of Theorem  1 %%%%%%%%%%%%%%%%%%%%%%%%%%
%%%%%%%%%%%%%%%%%%%%%%%%%%%%%%%%%%%%%%%%%%%%%%%%%%%%%%%%%%%%%%%%
\noindent {\bf Proof of Theorem \ref{thm:VfollowVREFandI}:} 
The steady-state tracking error is:
\[\lim_{s\rightarrow 0}s\bigg[\big[V_{refk}(s) - R_{vk}I_k(s)\big] - \big[G_kV_{refk}(s) - Z_kI_k(s)\big]\bigg]\]
\[= \lim_{s\rightarrow 0}s(1-G_k)V_{refk}(s) - \lim_{s\rightarrow 0}s(R_{vk} - Z_k)I_k(s)\]
Now, $G_k$ formulated in \eqref{eq5} can be written as:
\[\begin{array}{lll}
G_k = \frac{G_{pk}K_{volk}}{1+G_{pk}K_{volk}}
\end{array}\]
\[\textrm{where }\begin{array}{lll} G_{pk}=\frac{1}{C_ks}\frac{G_{ck}K_{curk}}{1+G_{ck}K_{curk}}, \hspace*{2mm} G_{ck}=\frac{1}{L_ks+R_k}\end{array}\]
Suppose $K_{volk}$ has  poles at $\pm j\omega_o.$ Hence, $K_{volk}$ can be factored as $K_{volk}=\frac{n_{volk}}{d_{volk}(s^2+\omega_o^2)}.$
Thus, because $G_{pk}$ has a pole at $s = 0$, it is evident that  $(1-G_k)$ will have a numerator with a factor $(s^2+\omega_o^2)s.$ This implies that $G_k(j\omega_o)=1.$
Also, from the final value theorem, it follows that:
\[\lim_{s\rightarrow 0}s (1-G_{k}) \frac{\omega_o}{s^2+\omega_o^2}=0.\]
In a similar manner, we can rewrite the expression of $Z_k$ as:
 \[\begin{array}{lll}
 Z_k=G_kR_{vk}+(1-G_k)(\frac{1}{sC_k})\frac{1}{1+G_{ck}K_{curk}}
 \end{array}\]
Since $G_k(j\omega_o)=1$, it follows that $Z_k(j\omega_o)=R_{vk}.$
Also, from the final value theorem, it follows that:
\[\begin{array}{lll}
\centering
& &\lim_{s\rightarrow 0}s(R_{vk}-Z_k)\frac{\omega_o}{s^2+\omega_o^2} \\
&=&\lim_{s\rightarrow 0}s(1-G_k)R_{vk}\frac{\omega_o}{s^2+\omega_o^2}\\  
& & +\lim_{s\rightarrow 0}s(1-G_k)(\frac{1}{sC_k})\frac{1}{1+G_{ck}K_{curk}}\frac{\omega_o}{s^2+\omega_o^2} = 0
\end{array}\]
As the term $(1-G_k)$ has a factor $s(s^2+\omega_o^2)$ it follows that 
\[\lim_{s\rightarrow 0}s(1-G_k)R_{vk}\frac{\omega_o}{s^2+\omega_o^2}=0\]
\[\lim_{s\rightarrow 0}s(1-G_k)(\frac{1}{sC_k})\frac{1}{1+G_{ck}K_{curk}}\frac{\omega_o}{s^2+\omega_o^2}=0\] 

This implies that in steady state the output of inner-loop controller tracks the input sinusoidal reference $v_{rk}(t) = v_{refk}(t) - R_{vk}i_k(t)$ with zero error.\\
\vspace{-0.78cm}
\subsection{Extension for Multiple Inverter System}
The closed-loop averaged output voltage dynamics of \eqref{eq5} can be extended for $N$ inverters and described by:
\begin{equation}\label{eq6}
\underline{\mathbf{V}}(s) = \mathbf{G}\underline{\mathbf{V}}_{\mathbf{ref}}(s) - \mathbf{Z}\underline{\mathbf{I}}(s)
\end{equation}
where, $\underline{\mathbf{V}}(s) \triangleq [V_1(s) \hspace*{1mm} V_2(s) \ldots V_N(s)]^T$, $ \underline{\mathbf{V}}_{\mathbf{ref}} \triangleq [V_{ref1}(s) \hspace*{1mm} V_{ref2}(s) \ldots V_{refN}(s)]^T$, $\underline{\mathbf{I}}(s) \triangleq [I_1(s) \hspace*{1mm} I_2(s) \ldots I_N(s)]^T$, $\mathbf{G} \triangleq \text{diag}(G_1, G_2, \ldots G_N)$ and  $\mathbf{Z} \triangleq \text{diag}(Z_1, Z_2, \ldots Z_N)$. 
% Considering the network in Fig. \ref{multiinverter} with $N$ branch resistances $r_k$ for $k = 1, 2, \ldots N$ and common load $Z_L$ at PCC, the voltage across and current through $Z_L$ can be formulated as
% \begin{comment}
%  \begin{align}\label{vloadiload}
%  \end{align}
% \end{comment}
% \[V_L(s) = h(s) \sum_{k=1}^N \frac{V_{k}(s)}{r_k}, \hspace*{4mm}
% I_L(s) = \frac{h(s)}{Z_L(s)} \sum_{k=1}^N \frac{V_{k}(s)}{r_k}\]
% where, $h(s)^{-1} := Z_L(s)^{-1}+\sum_{k=1}^N r_k^{-1}$. 
Assuming a linear load, the closed-loop expression of $\underline{\mathbf{V}}(s)$ and $\underline{\mathbf{I}}(s)$ can be formulated by using the admittance matrix $\mathbf{Y}(s)$, as follows:
\begin{subequations}\label{eq8}
\begin{align}
\label{eq8a}
\underline{\mathbf{V}}(s) &= [I-\mathbf{Z}\mathbf{Y}(s)]^{-1}\mathbf{G}\underline{\mathbf{V}}_{\mathbf{ref}}(s) \\ \label{eq8b}
\underline{\mathbf{I}}(s) &= \mathbf{Y}(s)[I-\mathbf{Z}\mathbf{Y}(s)]^{-1}\mathbf{G}\underline{\mathbf{V}}_{\mathbf{ref}}(s)
\end{align}
\end{subequations}
where $\mathbf{Y}(s) := \mathbf{\Lambda} - h(s)\pmb{\lambda} \pmb{\lambda}^T$, $h(s)^{-1} := Z_L(s)^{-1}+\sum_{k=1}^N r_k^{-1}$, 
% \begin{equation}\label{eq7}
% \underline{\mathbf{I}}(s) = \mathbf{Y}(s)\underline{\mathbf{V}}(s) = (\mathbf{\Lambda} - h(s)\pmb{\lambda} \pmb{\lambda}^T)\underline{\mathbf{V}}(s)   
% \end{equation}
$\pmb{\lambda}=[r_1^{-1} \hspace*{1mm} r_2^{-1} \ldots r_N^{-1}]^T$, and $\mathbf{\Lambda}=\text{diag}(r_1^{-1}, r_2^{-1}, \ldots r_N^{-1})$. In a similar way, $\underline{\mathbf{V}}_{\mathbf{ref}}$ can be formulated by rewriting in matrix form:
\begin{equation}\label{vref}
\underline{\mathbf{v}}_{\mathbf{ref}} = [E^{*}\mathbf{1} - \mathcal{N}(\underline{\mathbf{P}} - \underline{\mathbf{P}}^{*})]\sin \omega_o t = \mathbf{E}\sin \omega_o t   
\end{equation}
where $\mathcal{N}  \triangleq \text{diag}(n_1, n_2, \ldots n_N)$, $\underline{\mathbf{P}}  \triangleq [P_1 \hspace*{1mm} P_2 \ldots P_N]^T$, $\underline{\mathbf{P}}^{*}  \triangleq [P_1^* \hspace*{1mm} P_2^* \ldots P_N^*]^T$.\\
%%%%%%%%%%%%%%%%%%%%%%%%%%%%%%%%%%%%%%%%%%%%%%%%%%%%%%%%%%%%%%%%
%%%%%%%%%%%%%%%%% Proof of Theorem  2 %%%%%%%%%%%%%%%%%%%%%%%%%%
%%%%%%%%%%%%%%%%%%%%%%%%%%%%%%%%%%%%%%%%%%%%%%%%%%%%%%%%%%%%%%%%
{\bf Proof of Theorem \ref{thm:noQflowRload}:} For clarification, $\mathbf{R_v} := diag(R_{v1} \hspace*{1mm} R_{v2} \ldots R_{vN})$.
Also, $\underline{\mathbf{V}}_{\mathbf{ref}}$ can be written as $\bm{E}\mu(s)$, where $\mu(s)=\frac{\omega_o}{s^2+\omega_o^2}$ is the Laplace transform of a sinusoid with frequency $\omega_o$.
% \begin{eqnarray}
% \underline{\mathbf{V}}(s) &=& \big(I-\mathbf{R_v}\mathbf{Y}(s)\big)^{-1}\bm{E}\mu(s)\\
% \underline{\mathbf{I}}(s) &=& \mathbf{Y}(s)\big(I-\mathbf{R_v}\mathbf{Y}(s)\big)^{-1}\bm{E}\mu(s)
% %V_L(s)&=& R(s) \bm{\lambda}^T \underline{\mathbf{V}}(s)\\
% %I_L(s)&=& \frac{\underline{\mathbf{V}}(s)}{Z_L(s)}=\frac{R(s)}{Z_L(s)}\bm{\lambda}^T\underline{\mathbf{V}}(s)
% \end{eqnarray}
Using the Woodbury matrix identity and matrix manipulation:
% \begin{comment}
% i.e. $(\mathbf{W}+\mathbf{UCV})^{-1}=\mathbf{W}^{-1}-\mathbf{W}^{-1}\mathbf{U}(\mathbf{C}^{-1}+\mathbf{VW}^{-1}\mathbf{U})^{-1}\mathbf{VW}^{-1}.$\\
% Using $\mathbf{W}=\bLam$, $\mathbf{C}=1,\ \mathbf{V}=-\blam^T,$ $\mathbf{U}=h(s)\blam$ for $\mathbf{Y}(s)=(\bm{\Lambda}-h(s)\bm{\lambda}\bm{\lambda}^T)$,
% \end{comment}
\begin{eqnarray*}
{\mathbf{Y}(s)}^{-1}
% \begin{comment}
% &=&\big(\bm{\Lambda}-h(s)\blam\blam^T\big)^{-1}\\
% &=&\bLam^{-1}+\bLam^{-1}h(s)\blam\big(1-h(s)\blam^{T}\bLam^{-1}\blam\big)^{-1}\blam^{T}\bLam^{-1}\\
% \end{comment}
% &=&\bLam^{-1}+\frac{h(s)}{1-h(s)
% \blam^T\bLam^{-1}\blam}\bLam^{-1}\blam\blam^{T}\bLam^{-1}\\
% \begin{comment}
% &=&\bLam^{-1}+\frac{h(s)}{1-h(s)
% \blam^T\bLam^{-1}\blam}\bone\bone^{T}\\
% \end{comment}
&=&\bLam^{-1}+\eta\bone\bone^{T},\ \ \eta:={h(s)}{1-h(s)
\blam^T\bLam^{-1}\blam}^{-1}
\end{eqnarray*}
\begin{eqnarray*}
\text{Therefore,} \hspace*{2mm}
\underline{\mathbf{I}}(s)&=&(\mathbf{Y}(s)^{-1}+\mathbf{R_v}I)^{-1}\bm{E}\mu(s)\\
% &=&(\Lambda^{-1}+\eta \bm{1}\bm{1}^T+R_vI)^{-1}\bm{E}\mu(s)\\
&=&\big(\underbrace{(\Lambda^{-1}+\mathbf{R_v}I)}_{\bLam_v^{-1}}+\eta \bm{1}\bm{1}^T\big)^{-1}\bm{E}\mu(s)\\
% &=& (\bLam_v^{-1}+\eta\bm{1}\bm{1}^T)^{-1}\bm{E}\mu(s)\\
% &=& [\bLam_v-\eta\frac{\bLam_v\bm{1}\bm{1}^T\bLam_v)}{1+\eta \bm{1}^T\bLam_v\bm{1}}] \bm{E}\mu(s)\\
&=& [\bLam_v-\eta\frac{\bLam_v\bm{1}\bm{1}^T\bLam_v)}{1+\eta \bm{1}^T\bLam_v\bm{1}}] \bm{E}\mu(s)\\
% \begin{comment}
% &=&\bLam_v\bE\mu(s)-\alpha \bLam_v\bone\bone^T\bLam_v\bE\mu\\
% \end{comment}
&=& \bLam_v\bE\mu(s)-\alpha \blam_v\blam_v^T\bE\mu\\
\textrm{where}\hspace*{2mm}{\alpha^{-1}}
% \begin{comment}
% &=&\frac{1+\eta\bone^T\bLam_v\bone}{\eta}
% \end{comment}
&=& Z_L^{-1}+\sum _m ({r_m+R_{vm}})^{-1}\\
\blam_v &=&\big(\begin{array}{lll}\frac{1}{r_1+R_{v1}}&\ldots&\frac{1}{r_N+R_{vN}}\end{array}\big)^T\\ \bLam_v &=&diag\bigg(\frac{1}{r_1+R_{v1}},\ldots,\frac{1}{r_N+R_{vN}}\bigg)
\end{eqnarray*}
Thus, it follows that for the $k^{th}$ inverter in the steady state:
\[\begin{array}{lll}
i_k(t) &=& I_k\sin(\omega_o t+\phi_k)\\
\end{array}
\]
 where:
 \[I_k = \bigg(\sum_{m=1}^N \frac{E_m}{r_m+R_{vm}}\bigg)\frac{\gamma_k}{r_k+R_{vk}}\]
\[\gamma_k=\sqrt{\big(\beta_k-|\alpha|\cos\angle \alpha\big)^2+|\alpha^2|\sin^2\angle\alpha}\] \[\beta_k={E_k}({\sum_{m=1}^N {E_m}({r_m+R_{vm}})^{-1})^{-1}}\]
\[\cos\phi_k=({\beta_k-|\alpha(j\omega_o)|\cos\angle \alpha(j\omega_o)})/{\gamma_k};\]
\[\sin \phi_k=({-|\alpha(j\omega_o)|\sin\angle \alpha(j\omega_o)})/{\gamma_k};\]
By denoting $|\alpha|$ = $|\alpha(j\omega_o)|$ and $\angle \alpha$ = $\angle \alpha(j\omega_o)$,
\[\tan \phi_k=\displaystyle \frac{-|\alpha|\sin\angle \alpha}{\beta_k-|\alpha|\cos\angle \alpha(j\omega_o)}=\displaystyle \frac{\sin\angle \alpha}{\cos\angle \alpha -\frac{\beta_k}{|\alpha|}}\]
Now, from Theorem \ref{thm:VfollowVREFandI}:
\[\begin{array}{lll}
v_k(t) &= V_{refk}(t) - R_{vk}i_k(t)
% &= E_k\sin\omega_ot - R_{vk}I_k\sin(\omega_o t+\phi_k)\\
% \begin{comment}
% &= (E_k - R_{vk}I_k\cos\phi_k)\sin\omega_ot - R_{vk}I_k\sin\phi_k\cos\omega_ot\\
% \end{comment}
= E_k\sin(\omega_ot+\psi_k)
\end{array}\]
\[\textrm{where }V_k = \sqrt{(E_k - R_{vk}I_k\cos\phi_k)^2 + (R_{vk}I_k\sin\phi_k)^2}\]
\[\textrm{ and }\tan\psi_k = \frac{-R_{vk}I_k\sin\phi_k}{E_k - R_{vk}I_k\cos\phi_k}=\frac{\sin\phi_k}{\cos\phi_k - \frac{E_k}{R_{vk}I_k}}\]\\
%%%%%%%%%%%%%%%%%%%%%%%%%%%%%%%%%%%%%%%%%%%%%%%%%%%%%%%%%%%%%%%%
%%%%%%%%%%%%%%%%% Proof of Theorem  3 %%%%%%%%%%%%%%%%%%%%%%%%%%
%%%%%%%%%%%%%%%%%%%%%%%%%%%%%%%%%%%%%%%%%%%%%%%%%%%%%%%%%%%%%%%%
{\bf Proof of Theorem \ref{thm:diffInPhaseOfCurrents}:} From the expression of $\tan\phi_k$ deduced in the previous proof. Now, $\tan(\phi_k-\phi_j) $:
\[\begin{array}{lll}
&= \dfrac{\dfrac{\sin\angle \alpha}{\cos\angle \alpha-{\beta_k}/{|\alpha|}}-\dfrac{\sin\angle \alpha}{\cos\angle \alpha-{\beta_j}/{|\alpha|}}}{1+\displaystyle\prod_{\substack{
 l=k,j}}\bigg(\dfrac{\sin\angle \alpha}{\cos\angle \alpha-{\beta_l}/{|\alpha|}}\bigg)}
 \end{array}\]
 \[\begin{array}{lll}
&= \dfrac{\big({(\beta_k-\beta_j)}/{|\alpha|}\big)\sin \angle \alpha}{\displaystyle\prod_{\substack{
 l=k,j}} \bigg(\cos\angle \alpha -{\beta_l}/{|\alpha|}\bigg)+\sin^2\angle \alpha}
  \end{array}\]
 \[\begin{array}{lll}
&=\dfrac{{(\delta_k-\delta_j)\sin\angle \alpha}/({|\alpha| \blam_v^T(\bone+\bdelta)})}{\displaystyle\prod_{\substack{
 l=k,j}} \bigg(\cos\angle \alpha -{\beta_l}/{|\alpha|}\bigg)+\sin^2\angle \alpha}
  \end{array}\]
 \[\begin{array}{lll}
&= \dfrac{{(\delta_k-\delta_j)\sin\angle \nu}/({|\nu| (1+\bxi^T\bdelta)})}{\displaystyle\prod_{\substack{
 l=k,j}}\bigg(\cos\angle \nu-\dfrac{1+\delta_l}{(1+\bxi^T\bdelta)|\nu|}\bigg)+\sin^2\angle \nu} \vspace*{1mm}\\
&=\dfrac{\operatorname{Im}(\nu){(\delta_k-\delta_j)}/{ (1+\bxi^T\bdelta)}}{\displaystyle\prod_{\substack{
 l=k,j}} \bigg(\operatorname{Re}(\nu)-\dfrac{1+\delta_l}{1+\bxi^T\bdelta}\bigg)     +\operatorname{Im}^2(\nu)}\\
\end{array}\]
where $\beta_k={E_k}/{\blam_v^T\bE}={E^*(1+\delta_k)}/{\blam_v^T(1+\bdelta)E^*}, \bxi={\blam_v}/{\blam_v^T\bone},\hspace*{2mm} \nu:=\alpha\blam_v^T\bone\}$. Note that:  
\[\operatorname{Im}(\nu)=\operatorname{Im}(\nu -1)\approx -\operatorname{Im}\Big({1}/{Z_L\blam_v^T\bone}\Big)\]
\[1-\operatorname{Re}(\nu)=\operatorname{Re}\left({1}/{Z_L\blam_v^T\bone}\right)\]

\ifCLASSOPTIONcaptionsoff
  \newpage
\fi
\bibliographystyle{IEEEtran}
\bibliography{biblio} 

\end{document}